\RequirePackage{amsmath}
\RequirePackage[dvipsnames]{xcolor}
\documentclass[runningheads]{llncs}
\usepackage[T1]{fontenc}
\usepackage{graphicx}
\usepackage{algorithm}
\usepackage[noEnd,indLines,rightComments=false]{algpseudocodex}
\usepackage{amssymb}
\usepackage{cite}
\usepackage{hyperref}
\usepackage{mathtools}
\usepackage{multirow}
\usepackage{paralist}

\usepackage{tikz}
\usetikzlibrary{shapes.geometric, arrows, positioning}
\usepackage{xspace}
\usepackage{upgreek}

\usepackage[capitalize,english,nameinlink]{cleveref}
\crefalias{AlgoLine}{line}
\crefname{line}{\text{line}}{\text{lines}} 
\crefname{item}{\text{item}}{\text{items}} 
\crefname{example}{\text{Example}}{\text{Examples}} 
\crefname{assumption}{\text{Assumption}}{\text{Assumptions}} 
\crefname{algorithm}{\text{Algorithm}}{\text{Algorithms}}

\newcommand{\Dom}{\mathcal{D}}
\newcommand{\States}{\mathcal{X}}

\newcommand{\Hyps}{\mathcal{H}}
\newcommand{\HypsPD}{\Hyps_\mathrm{pd}}
\newcommand{\HypsSTB}{\Hyps_\mathrm{dec}}
\newcommand{\Reals}{\mathbb{R}}
\newcommand{\PosReals}{\Reals_{> 0}}

\newcommand{\norm}[1]{\left\lVert #1 \right\rVert}
\newcommand{\INT}{\mathrm{int}}

\newcommand{\ve}{\mathbf{e}}
\newcommand{\vx}{\mathbf{x}}
\newcommand{\vu}{\mathbf{u}}
\newcommand{\vy}{\mathbf{y}}
\newcommand{\vpara}{{\boldsymbol{\uptheta}}}
\newcommand{\cand}{\ensuremath{\Bar{\vpara}}}

\newcommand{\lya}{\ensuremath{V}}
\newcommand{\ctrl}{\ensuremath{\kappa}}
\newcommand{\gradlya}{\ensuremath{\nabla\lya}}
\newcommand{\derlya}{\ensuremath{\gradlya(\vx)}}
\newcommand{\fbb}{\ensuremath{f}}
\newcommand{\liederlya}{\derlya\cdot\fbb(\vx)}

\newcommand{\dist}{\delta}
\newcommand{\diff}{\mu}
\newcommand{\rob}{\gamma}
\newcommand{\lipbb}{\ensuremath{L}}
\newcommand{\lipkk}{\ensuremath{L_{\ctrl}}}
\newcommand{\LieUB}{\ensuremath{\mathit{LieUB}}}

\newcommand{\zero}{\mathbf{0}}

\newcommand{\xj}{{\Bar{\vx}}}

\newcommand{\yj}{{\Bar{\vy}}}

\newcommand{\Cover}{\mathcal{C}}
\newcommand{\Reg}{\mathcal{R}}
\newcommand{\lipreg}{\ensuremath{\lipbb_\Reg}}
\newcommand{\Ri}{\Reg_i}

\newcommand{\Ball}{\mathcal{B}}

\newcommand{\ROI}{\text{ROI}\xspace}
\newcommand{\CEGIS}{\text{CEGIS}\xspace}
\newcommand{\ACCPM}{\text{ACCPM}\xspace}
\newcommand{\CVXPY}{\text{CVXPY}\xspace}
\newcommand{\dReal}{\text{dReal}\xspace}

\begin{document}
\title{Certifying Lyapunov Stability of Black-Box Nonlinear Systems via Counterexample Guided Synthesis (Extended Version)}
\titlerunning{Lyapunov Stability of Black-Box Systems via CEGIS}
%

\author{Chiao Hsieh\inst{1}\orcidID{0000-0001-8339-9915} \and
Masaki Waga\inst{1}\orcidID{0000-0001-9360-7490} \and
Kohei Suenaga\inst{1}\orcidID{0000-0002-7466-8789}}
\authorrunning{C. Hsieh et al.}
%
\institute{
Graduate School of Informatics, Kyoto University, Kyoto, Japan\\
\email{\{chsieh16, mwaga, ksuenaga\}@fos.kuis.kyoto-u.ac.jp}}

\maketitle              
\begin{abstract}
Finding Lyapunov functions to certify the stability of control systems has been an important topic for verifying safety-critical systems.
Most existing methods on finding Lyapunov functions require access to the dynamics of the system.
Accurately describing the complete dynamics of a control system however remains highly challenging in practice.
Latest trend of using learning-enabled control systems further reduces the transparency.
Hence, a method for black-box systems would have much wider applications.

Our work stems from the recent idea of sampling and exploiting Lipschitz continuity to approximate the unknown dynamics.
Given Lipschitz constants, one can derive a \emph{non-statistical upper bounds} on approximation errors;
hence a strong certification on this approximation can certify the unknown dynamics.
We significantly improve this idea by directly approximating the Lie derivative of Lyapunov functions instead of the dynamics.
We propose a framework based on the learner--verifier architecture from Counterexample-Guided Inductive Synthesis~(\CEGIS).
Our insight of combining \emph{regional verification conditions} and \emph{counterexample-guided sampling} enables a guided search for samples to prove stability region-by-region.
Our \CEGIS algorithm further ensures termination.

Our numerical experiments suggest that it is possible to prove the stability of 2D and 3D systems with a few thousands of samples.
Our visualization also reveals the regions where the stability is difficult to prove.
In comparison with the existing black-box approach, our approach at the best case requires less than 0.01\% of samples.

\end{abstract}
\section{Introduction}\label{sec:intro}

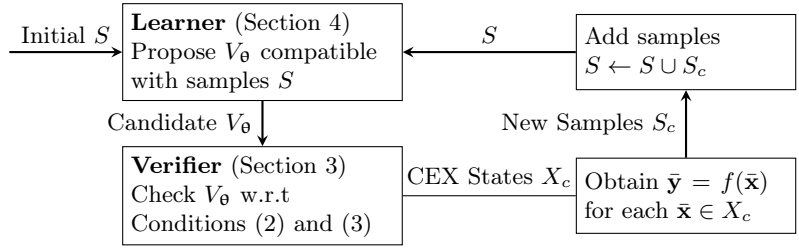
\begin{figure}[t]
    \centering
    \tikzstyle{process} = [rectangle, minimum width=15mm, minimum height=1cm, draw=black]
\tikzstyle{verify} = [rectangle, text width=28mm, minimum height=1cm, draw=black, fill=white]
\tikzstyle{decision} = [diamond, minimum width=15mm, minimum height=0.5cm, aspect=3, text centered, draw=black]
\tikzstyle{sum} = [circle, minimum width=0.25cm, minimum height=0.25cm, text centered, draw=black]
\tikzstyle{arrow} = [thick,->,>=stealth]

\pgfdeclarelayer{bg}
\pgfsetlayers{bg,main}

\begin{tikzpicture}[node distance=15mm,scale=1.0,every node/.style={transform shape}]
\node (init) {};
\node (learn)  [process, text width=35mm, right=15mm of init] {%
    \textbf{Learner} (\cref{sec:learner})\\ Propose $\lya_\vpara$ compatible with samples $S$};
\node (verify1) [verify, text width=35mm, below=.6cm of learn] {%
    \textbf{Verifier} (\cref{sec:approx})\\ Check $\lya_{\vpara}$ w.r.t \\Conditions~\eqref{cond:lya-pd} and~\eqref{cond:lya-stab}};
\node (sample) [process, text width=27mm, right=23mm of verify1] {%
    Obtain $\yj=\fbb(\xj)$ for each $\xj \in X_c$};
\node (add) [process, text width=27mm] at (sample |- learn) {%
    Add samples \\
    $S \gets S \cup S_c$};

\draw [arrow] (init) -- node [above] {Initial $S$} (learn);
\draw [arrow] (learn) -- node[left] {Candidate $\lya_\vpara$} (verify1);
\draw [-] (verify1) -- node[above] {CEX States $X_c$} (sample);
\draw [arrow] (sample) -- node[left] {New Samples $S_c$} (add);
\draw [arrow] (add) -- node [above] {$S$} (learn);

\end{tikzpicture}
    \caption{Architecture of \CEGIS of Lyapunov functions.}
    \label{fig:wbcegus-flow}
\end{figure}

Lyapunov method is a powerful tool for dynamical system analysis.
The existence of a Lyapunov function allows the study of important properties of the system, such as stability or positive invariants~\cite{khalil_nonlinear_2002}.
Even though the well-known Lyapunov theorems were proposed more than a century ago,
the importance of Lyapunov method has led to a large amount of research for automated discovery of Lyapunov functions for a given system~\cite{hafstein_review_2015}.
The major limitation of existing methods is that they usually assume the dynamical system is a \emph{white-box} model,
such as ordinary differential equations~(ODE).
In practice, however, models are rare and often a mixture of ODEs, simulation code, and observed data.
A method for \emph{black-box} systems that uses input, output, or partial information is more viable.

Up until recently,
Zhou et al.~\cite{zhou_neural_2022} proposed a data-driven method for stability and Lyapunov-based control for black-box systems.
The main idea of their method is to use \emph{evenly-spaced} samples to construct an approximation and synthesize a Lyapunov function with the approximation.
If the black-box system is Lipschitz continuous and the set of the samples is dense enough,
then the correctness of the Lyapunov function is formally guaranteed in a \emph{non-statistical} way.
However, their method requires an excessive number of samples to achieve the formal guarantee.
Indeed, for example, 9 million samples are used in \cite{zhou_neural_2022} for showing the stability of the Van der Pol oscillator, a second-order 2D system.

Our goal is to reduce the number of samples and still show the stability under the Lipschitz continuity assumption.
Our main insight is as follows:
Not every region in the state space requires the same density of samples.
We, therefore, aim to lazily sample the state space following the Counterexample-Guide Inductive Synthesis~(\CEGIS) framework~\cite{solar-lezama_program_2008}.
In the context of Lyapunov function synthesis, the \CEGIS framework is a search strategy for finding valid Lyapunov functions.
It is formalized as the interaction between a \emph{learner} and a \emph{verifier}.
The learner proposes a candidate function,
and the verifier checks if the candidate is a valid Lyapunov function for the given dynamics.
If the candidate is valid, then we found a Lyapunov function certifying the stability.
Otherwise, a counterexample falsifying the candidate is generated,
and the learner proposes a new candidate accounting for the counterexample.

\cref{fig:wbcegus-flow} describes the overall flow of our \CEGIS algorithm.
The crucial feature and the novelty of our \CEGIS is the \emph{regional verification condition} for our verifier to check the validity of a candidate with respect to a black-box system.
Our regional verification condition proves the formal stability through sampling the black-box system,
and it supports lazy sampling and avoids the need for evenly-spaced sampling.
This enables counterexample-guided sampling to obtain new samples only when necessary.
We further establish the theorems to show that lazy sampling can be as powerful as evenly-spaced sampling.
We also provide the termination of our \CEGIS algorithm:
Our \CEGIS algorithm either synthesizes or shows the absence of a Lyapunov function in the hypothesis space of candidates.

We implemented a prototype of our \CEGIS algorithm and evaluated it with benchmarks from the literature~\cite{abate_fossil_2021,zhou_neural_2022}.
Our experiment shows our prototype can synthesize a Lyapunov function with a few thousand of samples for 2D and 3D systems,
and it at best uses less than 0.01\% of samples compared with~\cite{zhou_neural_2022}.
We summarize our main contributions as follows:
\begin{compactenum}[1.]
    \item We propose a \CEGIS-based algorithm to synthesize Lyapunov functions for certifying the stability of black-box systems.
          Our verifier uses a novel black-box regional verification condition to check the validity of a candidate.
    \item We prove that our \CEGIS algorithm either synthesizes a Lyapunov function or shows the absence of a true Lyapunov function in the hypothesis space. 
         Our design of the hypothesis space for learning and the analytic center-based learner ensures the termination according to convex optimization theories.
    \item We implemented a prototype and evaluated our \CEGIS algorithm with existing benchmarks.
          The result demonstrates the effectiveness of our algorithm; it, at best, uses less than 0.01\% of samples compared with the existing work.
\end{compactenum}

\paragraph*{Paper Organization}
In \cref{sec:prelim}, we review dynamical systems, Lipschitz continuity, Lyapunov stability criteria, an overview of \CEGIS methods, and the convex feasibility problem.
In \cref{sec:approx}, we introduce our regional verification conditions for certifying Lyapunov stability for black-box systems.
In \cref{sec:learner}, we provide our choice of the hypothesis space and the learner design.
In \cref{sec:cegus}, we provide our \CEGIS algorithm and proof for its termination.
We discuss our experiments on nonlinear systems in \cref{sec:eval} 
and conclude in \cref{sec:conclusion}.

\subsection{Related Works}

\begin{table}
    \centering
    \caption{Comparison on \CEGIS of Lyapunov functions: ``BB'' stands for ``Black-Box Systems'', ``Term.'' stands for ``Termination'', and ``CS'' stands for ``Control Synthesis''.}
    \label{tab:related}
    \begin{tabular}{l|c|c|c}
Approaches
     &  BB  & Term. & CS \\ \hline
FOSSIL~\cite{ahmed_automated_2020,abate_formal_2021,abate_fossil_2021}
     &             &            &            \\
Chang et al.~\cite{chang_neural_2019}
     &             &            & \checkmark \\
Chen et al.~\cite{chen_learning_2021,chen_learning_2021-1}
     &             & \checkmark &            \\
Berger et al.~\cite{berger_learning_2022,berger_counterexample-guided_2023}
     &             & \checkmark &            \\
Masti et al.~\cite{masti_counter-example_2023}
     &             & \checkmark & \checkmark \\
Ravanbakhsh et al.~\cite{ravanbakhsh_counter-example_2015,ravanbakhsh_counterexample-guided_2015,ravanbakhsh_robust_2016,ravanbakhsh_learning_2019}
     &             & \checkmark & \checkmark \\
Zhou et al.~\cite{zhou_neural_2022}
     &  \checkmark &            & \checkmark \\
Ours &  \checkmark & \checkmark &            \\
    \end{tabular}
\end{table}

Finding Lyapunov functions for white-box systems has been an active research topic since the 1960s.
We refer readers to recent surveys for comprehensive reviews~\cite{hafstein_review_2015,dawson_safe_2023}.
Here, we provide a high-level comparison in \cref{tab:related} between our approach and others from two threads for the Lyapunov stability analysis: sampling-based approaches with formal guarantees and \CEGIS approaches.
In short, we propose a black-box \CEGIS approach with termination that has not been explored in previous works.

\paragraph{Sampling-Based Lyapunov Stability}
For white-box systems, \cite{kapinski_discovering_2013,bobiti_delta-sampling_2015} have first studied $\delta$-sampling and extended to prove the Lyapunov stability.
\cite{bobiti_automated-sampling-based_2018} further considered the negative definiteness of the Lie derivative and proposed non-evenly spaced sampling approach.
For black-box systems, \cite{zhou_neural_2022} is the only approach using sampling to ensure the formal stability to our knowledge.
In~\cite{zhou_neural_2022}, an approximation of the unknown dynamics is constructed via evenly-spaced sampling
(or \emph{$\dist$-sampling}~\cite{bobiti_delta-sampling_2015}),
with a rigorous bound on approximation errors.
They verify the approximated dynamics plus the error bound to certify the Lyapunov stability of the unknown dynamics.
We reduced the number of samples significantly compared with \cite{zhou_neural_2022} via 
CEGIS and lazy sampling.

\paragraph{\CEGIS of Lyapunov Functions.}
Besides the \CEGIS-based tool, FOSSIL~\cite{abate_fossil_2021,abate_formal_2021,ahmed_automated_2020},
we reviewed papers that studied termination of \CEGIS~\cite{ravanbakhsh_counter-example_2015,ravanbakhsh_counterexample-guided_2015,ravanbakhsh_robust_2016,ravanbakhsh_learning_2019,chen_learning_2021,chen_learning_2021-1,berger_learning_2022,berger_counterexample-guided_2023,masti_counter-example_2023}.
All approaches focus on white-box systems and cast \CEGIS as a cutting-plane method for solving instances of convex feasibility problems.
Depending on the hypothesis space and the system dynamics, the verifier can be implemented with different constraint solving engines such as Satisfiability Modulo Theories used in~\cite{abate_formal_2021}, Mixed Integer Quadratic Programming used in~\cite{chen_learning_2021-1}, Semidefinite Programming used in~\cite{ravanbakhsh_learning_2019}, etc.
Our approach applies convex feasibility for showing termination but for black-box systems.

\section{Preliminaries}\label{sec:prelim}

In this section, we recall preliminary notions, including the continuity and stability of dynamical systems in Section~\ref{subsec:dyn},
and we briefly review the existing Counterexample Guided Inductive Synthesis~(\CEGIS) framework in Section~\ref{subsec:wb-cegus}.

\paragraph*{Notations}
We use $x\in\Reals$ for the real numbers, $|x|$ for the absolute value of $x \in \Reals$, $\vx\in\Reals^n$ for a column vector $\vx$ in the $n$-dimensional Euclidean space,
$\vx^T$ for the transpose of $\vx$,
$\zero$ for the vector of zeros as the origin,
and $\ve_i\in\Reals^n$ as the $i$-th basis vector.
We also use $\vx=[x_1\dotsc x_n]^T$ for denoting elements explicitly and use $x_i$ for the $i$-th element in $\vx$,
and equivalently $\vx=\sum_{i=1}^n x_i\ve_i$.
We use $\vx\cdot \vy = \vx^T\vy = \sum_{i=1}^n x_i y_i$ for the inner product of two vectors $\vx, \vy \in \Reals^n$,
and $\norm{\vx} = \sqrt{\vx\cdot\vx} = \sqrt{\sum_{i=1}^n x_i^2}$ for the Euclidean norm of $\vx\in \Reals^n$.
When there is no ambiguity, we use the notation $(\vx, \vu) \in \Reals^{n+m}$ as the concatenation of two vectors $\vx\in\Reals^n$ and $\vu\in\Reals^m$.
Given a scalar function $\lya:\Reals^n \mapsto \Reals$,
the gradient of $\lya(\vx)$ is denoted as $\gradlya(\vx)=[\frac{\partial \lya(\vx)}{\partial x_1}\dotsc \frac{\partial \lya(\vx)}{\partial x_n}]^T$.
Additionally, we use the following terms to distinguish different kinds of sets:
a \emph{domain} $\Dom \subseteq \Reals^n$ is an \emph{open and connected subset},
and a \emph{region} $\Reg \subset \Reals^n$ is a \emph{compact and connected subset}.
The list of notations is available in \cref{appx:symbols}.

\subsection{Lipschitz Continuity and Lyapunov Stability}\label{subsec:dyn}
Throughout this work, we consider a nonlinear dynamical system
modeled with a vector field $\fbb:\Reals^n \to \Reals^n$
where $\Reals^n$ is the state space.
The closed-loop control system is a system of ordinary differential equations~(ODEs) of the form:
\begin{equation}\label{sys:closed-loop}
\dot{\vx} = \fbb(\vx)
\end{equation}
Without loss of generality, we assume the origin $\zero \in \Reals^n$ is an equilibrium of the closed-loop system
so that $\fbb(\zero) = \zero$,
and we do not assume a closed-form expression of $\fbb$.
The main objective is to certify the stability of System~\eqref{sys:closed-loop} in the sense of Lyapunov,
i.e., stability guarantees are established if we can find Lyapunov functions for System~\eqref{sys:closed-loop}.
We next present preliminaries on Lipschitz continuity and Lyapunov stability analysis.

\begin{definition}[Lipschitz continuity]
A function $f:\Reals^{n_1} \to \Reals^{n_2}$ is \emph{Lipschitz continuous} in a domain $\Dom \subseteq \Reals^{n_1}$ if there is a constant $L > 0$ such that, for any $\vx, \Bar{\vx} \in \Dom$,
$\norm{f(\vx) - f(\Bar{\vx})} \leq L \norm{\vx - \Bar{\vx}}$.
Note that $L$ need not be the smallest value,
and we call any such $L$ a \emph{Lipschitz bound in $\Dom$} throughout.
\end{definition}
The existence and uniqueness theorem~\cite[Theorem~3.1]{khalil_nonlinear_2002} states that the Lipschitz continuity of the right-hand side of System~\eqref{sys:closed-loop} guarantees a unique solution trajectory passing through a given initial state over a bounded time.
Hence, Lipschitz continuity is a widely accepted assumption.
Moreover, several methods have been proposed to estimate Lipschitz bounds for black-box systems~\cite{Wood1996EstimationOT}.
We further provide a simple extension for discussing different regions within $\Dom$.
\begin{definition}[Regional Lipschitz Bound]
Given a Lipschitz continuous function $\fbb$ in a domain $\Dom$ and a region $\Reg \subseteq \Dom$, $\lipreg$ is a \emph{regional Lipschitz bound} in $\Reg$ if $\lipreg$ is a Lipschitz bound for some domain $\Dom'$ such that $\Reg \subseteq \Dom' \subseteq \Dom$.
\end{definition}
By definition, a Lipschitz bound $\lipbb$ in $\Dom$ is always a Lipschitz bound in $\Reg$ because $\Reg \subseteq \Dom' \subseteq \Dom$, so we assume $\lipreg \leq \lipbb$.
In this paper, we assume that a Lipschitz bound $\lipbb$ is provided for $\fbb$ in the entire domain $\Dom \subset \Reals^n$,
and regional Lipschitz bounds $\lipreg$ for some regions $\Reg$ may be provided but not required.

We use \emph{Lyapunov functions} to certify the asymptotic stability of System~\eqref{sys:closed-loop}.
Specifically, we focus on a region of interest $\States \subseteq \Dom$.
\begin{definition}[Lyapunov Function for Asymptotic Stability]\label{def:lya}
Given a \emph{region of interest~(ROI)} $\States \subseteq \Dom \setminus \{\zero\}$ surrounding but excluding the origin,
a continuously differentiable function $\lya: \Reals^n \to \Reals$ is called a \emph{Lyapunov function} for System~\eqref{sys:closed-loop} if $\lya$ satisfies the following:
\begin{align}
    & \lya(\zero) = 0 \land \forall \vx \in\States, \lya(\vx) > 0 \label{cond:lya-pd}\\
    & \forall \vx \in\States, \liederlya < 0 \label{cond:lya-stab}
\end{align}
where $\derlya$ is the \emph{gradient vector} of $\lya$,
and $\liederlya$ is the \emph{Lie derivative} of $\lya$ along the flow of System~\eqref{sys:closed-loop}.
We call $\lya$ is \emph{(locally) positive definite} if $\lya$ satisfies Condition~\eqref{cond:lya-pd}
and $\lya$ is \emph{(locally) decreasing along trajectories} of System~\eqref{sys:closed-loop} if $\lya$ satisfies Condition~\eqref{cond:lya-stab}.
\end{definition}

When $\States = \Dom \setminus \{\zero\}$, Lyapunov's stability theorem~\cite[Theorem 4.1]{khalil_nonlinear_2002} states that
the existence of a Lyapunov function in Definition~\ref{def:lya} guarantees the \emph{asymptotic stability} of System~\eqref{sys:closed-loop}.
The primary insight is that if $\lya$ is positive definite and monotonically decreasing along the trajectories,
it must eventually approach its minimum value 0, which indicates that the system must reach equilibrium $\zero$.
Intuitively, it is useful to think of $\lya$ as a generalized
energy: If the system is always dissipating energy, then it will eventually come to rest.
In practice, we choose an \ROI $\States \subset \Dom$ excluding a small ball around the origin to address the numerical robustness~\cite{gao_numerically-robust_2019}.

\subsection{Counterexample Guided Synthesis of Lyapunov Functions}\label{subsec:wb-cegus}

One of the common approaches to show the Lyapunov stability of a system is via \emph{counterexample-guided inductive synthesis (\CEGIS)} of Lyapunov functions~\cite{chang_neural_2019,ravanbakhsh_learning_2019,chen_learning_2021-1,abate_formal_2021,zhou_neural_2022}.
\cref{fig:wbcegus-flow} outlines a typical \CEGIS algorithm.
Starting from an initial set of samples $S$, a \CEGIS algorithm proceeds as follows:
\begin{compactenum}[1.]
    \item \textbf{Learner} proposes a candidate $\lya_\vpara$ based on the current samples $S$.
    \begin{itemize}
        \item If \textbf{Learner} cannot find any $\lya_\vpara$, stop with no candidates found.
    \end{itemize}
    \item \textbf{Verifier} checks if $\lya_\vpara$ satisfies Conditions~\eqref{cond:lya-pd} and~\eqref{cond:lya-stab}.
    \begin{itemize}
        \item If true, return $\lya_\vpara$ as the Lyapunov function.
        \item If false, find \emph{counterexample states} $X_c$ to falsify Conditions~\eqref{cond:lya-pd} and~\eqref{cond:lya-stab}.
    \end{itemize}
    \item The algorithm samples the output $\yj=\fbb(\xj)$ for each state $\xj \in X_c$ to obtain new samples $S_c$, adds $S_c$ to $S$, and goes back to Step~1.
\end{compactenum}
Overall, the learner avoids previously falsified candidates by including counterexamples,
and the verifier ensures the correctness of the \CEGIS algorithm.

\paragraph{Learners in \CEGIS of Lyapunov Functions}

One standard approach to synthesizing a Lyapunov function candidate in \CEGIS is to fix a parameterized function template and find an appropriate parameter vector with respect to the current samples $S$.
For simplicity, we denote a Lyapunov function candidate as $\lya_\vpara$ and its gradient as $\gradlya_\vpara$ with the \emph{parameter vector} $\vpara \in \Reals^d$ (or \emph{weights} in machine learning literature).
We assume $\lya_\vpara(\zero) = 0$ for every parameter $\vpara$ by design.
Moreover, we assume that the proposed candidates are \emph{observation compatible}~\cite[Definition~5]{ravanbakhsh_learning_2019}.

\begin{definition}[Observation Compatibility]\label{def:compatible}
A candidate $\lya_\vpara$ is \emph{compatible with a set of samples} $S$ if
\(\lya_\vpara(\xj) > 0\) and \(\gradlya_\vpara(\xj) \cdot \yj < 0\)
for all $(\xj, \yj) \in S$.
\end{definition}
Observation compatibility can be considered as a weakening of Conditions~\eqref{cond:lya-pd} and~\eqref{cond:lya-stab} with respect to $\xj$ seen in $S$ instead of the entire $\States$.
It can be enforced either as hard constraints~\cite{ravanbakhsh_learning_2019,abate_formal_2021,chen_learning_2021-1} or as soft constraints such as the empirical Lyapunov risk in~\cite{chang_neural_2019,zhou_neural_2022}.

\paragraph{Verifiers in \CEGIS of Lyapunov Functions}
A verifier in \CEGIS decides if the given candidate $\lya_\vpara$ is truly a Lyapunov function or falsifies $\lya_\vpara$ with a \emph{counterexample state} $\vx\in\States$ for Conditions~\eqref{cond:lya-pd} or~\eqref{cond:lya-stab}.
Since the verifier must check Conditions~\eqref{cond:lya-pd} and~\eqref{cond:lya-stab} against all states in the ROI $\States$,
a major challenge arises when $\fbb$ in System~\eqref{sys:closed-loop} is black-box.
That is, an approximation is required to interpolate or extrapolate from the current samples $S$ to unobserved states in $\States$ to check Condition~\eqref{cond:lya-stab}.
Further, we need to derive a bound on approximation error
so that checking the approximation with the error bound ensures Condition~\eqref{cond:lya-stab}.

\paragraph*{Termination of \CEGIS Algorithms}
Due to the uncountable hypothesis space for $\vpara$, the \CEGIS algorithm may never terminate.
Existing works reduce the learning problem in a selected hypothesis space, e.g., linear combinations of monomials~\cite{ravanbakhsh_learning_2019} and positive definite matrices~\cite{chen_learning_2021-1},
to \emph{convex feasibility problems with a separating oracle},
and apply existing \emph{cutting-plane methods} to guarantee the termination.
We will introduce convex feasibility and a specific cutting-plane method in \cref{subsec:convex}.
We will provide our design of the learner and verifier for our \CEGIS algorithm
and the reduction to convex feasibility in \cref{sec:learner}.

\subsection{Convex Feasibility with Separating Oracle}\label{subsec:convex}

\newcommand{\va}{\mathbf{a}}

In this work, we obtain a terminating procedure for Lyapunov function synthesis by reducing the problem to the convex feasibility problem.
In short, \emph{convex feasibility} is to find a point inside a convex set,
which is one of the fundamental problems in convex optimization.
Here, we specifically review the convex feasibility problem based on \emph{a separating oracle}~\cite{goffin_complexity_1996,jiang_improved_2020}.
As pointed out in~\cite{goffin_complexity_1996},
the separating oracle allows us to implicitly represent a general convex set defined by a possibly infinite intersection of convex sets,
which is crucial for our learner in \cref{sec:learner}.
We first define the separating oracle and the convex feasibility.

\begin{definition}[Separating Oracle]
Let $\Gamma \subseteq \Reals^d$ be a convex set with a non-empty interior $\INT(\Gamma)\neq\emptyset$.
Given a point $\cand \notin \INT(\Gamma)$,
we represent a \emph{separating hyperplane} for $\cand$ and $\Gamma$ by a pair $(\va, b)$ of a vector $\va\in\Reals^d$ and a scalar $b\in\Reals$ such that
$\va \cdot \cand \geq b$ and
\(
\Gamma \subseteq \{\vpara \mid \va \cdot \vpara < b\}.
\)
A \emph{separating oracle} of $\Gamma$ either answers $\cand\in\INT(\Gamma)$ or generates a separating hyperplane
for $\cand\notin\INT(\Gamma)$.
\end{definition}
\begin{definition}[Convex Feasibility]\label{def:convex-feasibility}
Given a separating oracle for a convex set $\Gamma \subseteq \Reals^d$ contained in a unit hypercube,
the \emph{convex feasibility problem} is to
either find a point $\vpara \in \Gamma$ or prove that $\Gamma$ does not contain a ball of radius $\rob$.
\end{definition}

Among various cutting-plane for solving the convex feasibility using a separating oracle,
we use the \emph{analytic center cutting-plane method (ACCPM)}, which is simple but effective.
See, e.g.,~\cite{jiang_improved_2020} for a comprehensive comparison of recent algorithms.
The insight of \ACCPM is to iteratively propose the \emph{analytic center} of the polytope defined by 
separating hyperplanes
and shrinks the polytope by adding a new separating hyperplane whenever the analytic center is rejected.
\begin{definition}[Analytic Center of a Polytope]
Given a polytope $\Hyps\subseteq \Reals^d$ defined by $k$ halfspaces,
that is, $\Hyps = \bigcap_{i=1}^k\{\vpara \mid \va_i \cdot \vpara < b_i\}$,
the \emph{analytic center} of $\Hyps$ is the point $\vpara^* \in \Reals^d$ such that 
\(
\vpara^* = \arg\max_{\vpara\in\Reals^d} \sum_{i=1}^k \ln (b_i -\va_i \cdot \vpara_i)
\).
\end{definition}
We include here a bound on queries to the oracle for \ACCPM.
\begin{theorem}[{From~\cite[Theorem~6.6]{goffin_complexity_1996}}]\label{thm:accpm}
The analytic center cutting-plane method~(\ACCPM) solves the convex feasibility problem with
$k$ queries to the separating oracle as soon as $k$ satisfies
\(
\frac{\rob^2}{d} \geq \frac{\frac{1}{2} +2d\ln(1+\frac{k+1}{8d^2})}{2d+k+1}
\)
where $d$ and $\rob$ are the same as in \cref{def:convex-feasibility}.
\end{theorem}

Now to implement a separating oracle for convex feasibility,
we consider when $\Gamma$ is a subset of a simpler convex set, for example,
$\Gamma \subseteq \{\vpara\mid g(\vpara) < 0\}$ for a convex function $g$.
A \emph{subgradient} of $g$ is commonly used as a separating oracle.
\begin{definition}[Subgradient and Subdifferential]
Let $g : \Reals^d \to \Reals$ be a convex function,
a vector $\va \in \Reals^d$ is called a \emph{subgradient} of $g$ at a point $\cand \in \Reals^d$
if for any $\vpara \in \Reals^d$, we have $g(\vpara) \geq g(\cand) + \va \cdot (\vpara - \cand)$.
The set of all subgradients of $g$ at $\cand$, denoted by $\partial g(\cand)$, is called the \emph{subdifferential} of $g$ at $\cand$.
\end{definition}
\begin{proposition}[{Existence of Subgradients~\cite[Proposition~5.4.1]{bertsekas_convex_2009}}]
Given a convex function $g$, the subdifferential $\partial g(\cand)$ for every $\cand\in\Reals^d$ is nonempty.
\end{proposition}
\begin{remark}
Subgradient generalizes the definition of a gradient to nonsmooth points in convex functions,
and the gradient $\nabla g(\cand)$ at a differentiable point is a subgradient.
See, e.g.,~\cite[Section 5.4]{bertsekas_convex_2009} for how to compute the subgradients.
\end{remark}

\begin{proposition}[Separating Hyperplane by Subgradient]\label{prop:grad-plane}
Given a convex function $g : \Reals^d \mapsto \Reals$,
let $\Gamma \subseteq \{\vpara \mid g(\vpara) < 0\}$,
and a point $\cand \in \Reals^d$ such that $g(\cand) \geq 0$.
The pair $(\va, b)$ of the subgradient $\va \in \partial g(\cand)$ and the scalar $b=\va\cdot\cand-g(\cand)$ is a separating hyperplane for $\cand$ and $\Gamma$.
Further, if $g$ is linear, the pair $(\nabla g(\cand), 0)$ is a separating hyperplane for $\cand$ and $\Gamma$.
\end{proposition}


\section{Black-Box Regional Verification}\label{sec:approx}

We present an algorithm to verify the Lyapunov stability of black-box systems from samples based on Lipschitz bounds.
This algorithm is used as the verifier (i.e., the left below component of \cref{fig:wbcegus-flow}) in the \CEGIS of Lyapunov functions, later in \cref{sec:cegus}.
Our algorithm is based on the reduction to satisfiability checking of a certain verification condition encoding the Lyapunov stability.

In Section~\ref{subsec:eps-provable}, we first define a \emph{$\dist$-provably decreasing} Lyapunov candidate with respect to \emph{evenly spaced samples} given a distance parameter $\dist>0$.
In Section~\ref{subsec:regional-vc},
we further extend our idea for an arbitrary set of samples and define \emph{regional verification conditions}.
With the help of $\dist$-provability, we then show in Theorem~\ref{thm:grid-vc} that our regional verification condition is, in fact, as strong as Condition~\eqref{cond:lya-stab}.
In addition, it allows the use of fewer samples from the black-box dynamics compared with evenly spaced sampling.
This paves the way to the verifier by counterexample guided sampling in Section~\ref{subsec:cegar}.

\subsection{Verification Condition for \texorpdfstring{$\dist$-}{}Evenly Spaced Samples}\label{subsec:eps-provable}

Here, we present a verification condition, named $\dist$-provability, assuming evenly spaced samples, i.e., the ROI $\States$ is covered by $\delta$-balls around the samples.
Our $\dist$-provability can be considered as a reformulation of existing conditions via $\delta$-sampling and Lipschitz bounds in~\cite{kapinski_discovering_2013,bobiti_delta-sampling_2015,zhou_neural_2022}.
We restate the definitions and the theorem in our notations to compare with our more general verification condition for lazy sampling in \cref{subsec:regional-vc}.

The idea is that, with $\dist$-evenly spaced samples, we can always find an observed state $\vx_i$ within $\dist$-distance of any unobserved state $\vx$.
We can use the expression $\derlya \cdot \vy_i$ with the observed output $\vy_i=\fbb(\vx_i)$ to approximate the Lie derivative $\derlya \cdot \fbb(\vx)$ with a bound on the approximation error.
Then, we can simply require that the approximation plus the error is decreasing to ensure that the Lie derivative is decreasing along trajectories,

\begin{definition}[$\dist$-cover]
For $\dist>0$, a \emph{$\dist$-cover} of $\States$ is a finite set of states $\{\vx_1, \dotsc, \vx_N\} \subset \Dom$ (or in short $\{\vx_i\}$) such that $\States \subseteq \bigcup_{i=1}^N \Ball_\dist(\vx_i)$, or equivalently,
for any $\vx \in \States$ there is $\vx_i \in \{\vx_i\}$ satisfying $\norm{\vx - \vx_i} \leq \dist$.
\end{definition}
Note that $\vx_i$ is in $\Dom$ and not necessarily in $\States$.
Therefore, we can choose $\vx_i \notin \States$ to construct a cover for $\States$ more easily.

\begin{definition}[$\dist$-provability]\label{def:provability}
Given System~\eqref{sys:closed-loop} and $\States$,
we say that a continuously differentiable function $\lya$ is \emph{$\dist$-provably decreasing} along trajectories
if there exists a $\dist$-cover $\{\vx_i\}$ of $\States$ such that,
for all $\vx_i \in \{\vx_i\}$ and $\vy_i=\fbb(\vx_i)$,
\begin{equation}\label{cond:eps-cover}
   \forall \vx \in \Ball_\dist(\vx_i) \cap \States, \quad
   \derlya \cdot \vy_i < -2M\lipbb\dist
\end{equation}
where $\lipbb$ is a Lipschitz bound for $\fbb$ in $\Dom$,
and $M \geq \sup_{\vx\in \States}\norm{\derlya}$ is an upper bound on the norm of the gradient.
\end{definition}

\begin{theorem}\label{thm:eps-equiv}
Given System~\eqref{sys:closed-loop} and a compact region $\States \subseteq \Dom \setminus \{\zero\}$,
a function $\lya$ is decreasing along trajectories \emph{(Condition~\eqref{cond:lya-stab})} if and only if $\lya$ is $\dist$-provably decreasing for some $\dist>0$ \emph{(Condition~\eqref{cond:eps-cover})}.
\end{theorem}
\begin{proof}
A complete proof is available in \cref{appx:thm:eps-equiv}.
\end{proof}
It follows from \cref{thm:eps-equiv} that there is a sufficiently small $\dist>0$ to prove $\dist$-provability if a Lyapunov function exists.
This suggests that proof via a dense enough sampling always exists.

\subsection{Regional Verification Condition for Arbitrary Samples}\label{subsec:regional-vc}

Although \cref{def:provability} allows us to prove the Lyapunov stability by sampling, the number of evenly spaced samples tends to be huge.
Therefore, we generalize \cref{def:provability} for arbitrary samples based on the following observations:
\begin{inparaenum}[(1)]
    \item The necessary density of the samples is not uniform over the ROI $\States$, and unevenly spaced sampling can be significantly more efficient;
    \item The nearest sample does not always provide the tightest upper bound, and the use of more than one sample can improve the approximation.
\end{inparaenum}
Based on these observations, we show
a less conservative condition for proving Lyapunov stability with samples.

Notice that if we allow any cover of $\States$, we no longer have a canonical mapping from any state $\vx \in \States$ to some center state $\vx_i$ provided by a $\dist$-cover.
Instead, we associate a set $S$ of samples to each region $\Reg$ and show that $\lya$ is decreasing along trajectories for each $\Reg$.
Furthermore, we show that it is sufficient to focus on the regions defined as a convex hull of sampled states.

The proof sketch is as follows:
Definition~\ref{def:regional-vc} provides the regional verification condition for a region using multiple samples.
Proposition~\ref{prop:bbox-regional} then proves that the regional verification conditions for all regions are a sufficient condition for Condition~\eqref{cond:lya-stab}.
Theorem~\ref{thm:grid-vc} further shows the equivalence of regional verification conditions and $\dist$-provability.
Finally, Theorem~\ref{thm:provable-convex} is our specialized theorem for using convex hulls of sampled states to cover $\States$.

\begin{proposition}\label{prop:lya-regional}
Let a set of regions $\Cover=\{\Ri\}_{i=1\dotsc N}$ be a cover of the \ROI $\States$,
i.e., $\States\subseteq \bigcup_{i=1}^N\Ri$.
If a function $\lya$ satisfies the following for all regions $\Ri \in \Cover$:
\begin{equation}\label{cond:regional-lya}
\forall \vx\in\Ri\cap\States,\quad \liederlya < 0,
\end{equation}
then $\lya$ satisfies Condition~\eqref{cond:lya-stab}.
\end{proposition}
\begin{remark}\label{remark:cover}
If $\{\Ri\}_{i=1\dotsc N}$ both covers and partitions $\States$, the converse is also true.
However, using a partition may add heavy burdens to the verifier due to non-convex $\States$ and complements of regions.
If we allow $\Ri \not\subseteq \States$ and $\Ri \cap \Reg_j \neq \emptyset$,
we can use regions of simpler shapes for verification.
We will discuss this in more detail in \cref{subsec:cegar} for efficient implementations.
\end{remark}

\begin{definition}[Regional Verification Condition]\label{def:regional-vc}
Given a region $\Reg \subseteq \Dom$,
a Lipschitz bound $\lipreg$ for $\fbb$ in $\Reg$,
a function $\lya: \Reals^n \to \Reals$,
and a sample $(\xj, \yj)$ with $\xj \in \Reg$ and $\yj = \fbb(\xj)$,
we define a function, $\LieUB_{\xj, \yj}: \Reals^n \to \Reals$ as:
\[
\LieUB_{\xj, \yj}(\vx) := \norm{\derlya}\lipreg\norm{\vx-\xj} + \derlya \cdot \yj
\]
For a set of samples $S = \{(\xj_j, \yj_j)\}_{j}$ where $\xj_j \in \Reg$ and $\yj_j=\fbb(\xj_j)$ for each $j$,
the function $\lya$ is \emph{regionally decreasing} in $\Reg$ witnessed by samples $S$ if
\begin{equation}\label{cond:regional-lya-bb}
\forall \vx \in\Reg\cap\States, \bigvee_{(\xj, \yj)\in S} \LieUB_{\xj, \yj}(\vx) < 0
\end{equation}
We also refer to Condition~\eqref{cond:regional-lya-bb} as the \emph{regional verification condition} throughout.
\end{definition}

\begin{proposition}\label{prop:bbox-regional}
Given a region $\Reg \subseteq \Dom$ and a Lipschitz bound $\lipreg$ for $\fbb$ in $\Reg$,
if a function $\lya$ is regionally decreasing in $\Reg$ witnessed by samples $S$,
then $\lya$ satisfies Condition~\eqref{cond:regional-lya} in $\Reg$.
\end{proposition}
\begin{proof}
The proof is to show that, for each sample $(\xj, \yj)$,
$\LieUB_{\xj, \yj}(\vx)$ is an upper bound of $\liederlya$.
Hence, any upper bound less than zero suffices to prove that $\lya$ is regionally decreasing.
A complete proof is available in \cref{appx:prop:bbox-regional}.
\end{proof}

\begin{theorem}[Equivalent Power to $\dist$-provability]\label{thm:grid-vc}
A function $\lya$ is $\dist$-provably decreasing for some $\dist$
if and only if
there exists a cover $\Cover=\{\Ri\}_{i=1\dotsc N}$ and a set of sample sets $\{S_i\}_{i=1\dotsc N}$, such that $\lya$ is regionally decreasing in every region $\Ri$ witnessed by $S_i$.
\end{theorem}
\begin{proof}
On a high level, the ``if'' direction holds because Condition~\eqref{cond:regional-lya-bb} implies Condition~\eqref{cond:regional-lya}, then Condition~\eqref{cond:lya-stab}, and thus Condition~\eqref{cond:eps-cover} by \cref{thm:eps-equiv}.
To prove the ``only if'' direction is to show Condition~\eqref{cond:eps-cover} with a $\dist$-cover implies Condition~\eqref{cond:regional-lya-bb} by construction.
A complete proof is available in \cref{appx:thm:grid-vc}.
\end{proof}
\cref{thm:grid-vc} states that, if a $\dist$-cover proof exists,
then there exists a proof using our regional verification conditions, and vice versa.
Our next theorem further relates the radius $\dist$ and the diameter of the region built from sampled states.

\begin{theorem}\label{thm:provable-convex}
If a function $\lya$ is $\dist$-provably decreasing for some $\dist > 0$,
then $\lya$ must be regionally decreasing in any region $\Reg \subseteq \Dom$ witnessed by samples $S$ satisfying the following two requirements:
\begin{itemize}
    \item The region $\Reg$ is the convex hull of the sampled states in $S$.
          That is,  $\Reg=\mathrm{conv}(\mathcal{V})$,
          where $\mathcal{V} = \{\xj\mid (\xj, \yj) \in S\}$.
    \item The diameter of $\Reg$ is bounded by:
\(
\mathrm{diam}(\Reg) \leq \frac{\dist}{2}\sqrt{\frac{2(n+1)}{n}}
\)
\end{itemize}
\end{theorem}
\begin{proof}
The high level idea is as follows:
Given any $\delta$-cover and any convex hull $\Reg$,
once the diameter of a convex hull is small enough,
the distances of all states in $\Reg$ to the center of the closest $\dist$-ball is also small,
so the $\delta$-provability should ensure that $\lya$ is regionally decreasing in $\Reg$.
The proof is available at \cref{appx:thm:provable-convex}.
\end{proof}
\cref{thm:grid-vc} and~\ref{thm:provable-convex} suggest covering $\States$ with convex hulls of samples instead of $\dist$-balls.
Next we discuss our covering strategy using convex hulls in \cref{subsec:cegar}.

\subsection{Regional Lyapunov Verification Algorithm}\label{subsec:cegar}

In this section, our goal is to verify if a candidate $\lya$ is a true Lyapunov function using regional verification conditions~(Condition~\eqref{cond:regional-lya-bb}) for regions in a cover of $\States$.
Recall \cref{thm:eps-equiv}:
A Lyapunov function is $\dist$-provable for a small enough $\dist$.
Thus, there are two cases when $\lya$ does not satisfy Condition~\eqref{cond:regional-lya-bb}:
\begin{compactenum}[I)]
    \item\label{case:true} $\lya$ is not a Lyapunov function.
    \item\label{case:spur} The cover of $\States$ is not fine enough for proof.
\end{compactenum}
Case~\ref{case:true} requires a counterexample for falsification
while Case~\ref{case:spur} requires searching for a fine enough cover of $\States$.
Moreover, as discussed in \cref{prop:lya-regional},
a cover of $\States$ with little overlap between regions is preferred.
Our approach combines Delaunay triangulation~\cite[Chapter 27]{toth_handbook_2017} over sampled states with counterexample-guided sampling to achieve the following desirable features:
\begin{inparaenum}[1)]
    \item It is fast to retrieve samples used in Condition~\eqref{cond:regional-lya-bb} in a region.
    \item Each region is of a simple shape for fast verification.
    \item The generated regions overlap as little as possible.
    \item It finds a counterexample for Case~\ref{case:true}.
    \item It incrementally finds \emph{a finer cover}\footnote{Here we mean a cover using smaller regions and not strictly a refinement of a cover.}  of $\States$ for Case~\ref{case:spur}.
\end{inparaenum}
Overall, our Lyapunov verification algorithm contains three components:
\begin{compactitem}
    \item Construction of regions from the current set of samples
    \item Verification of Condition~\eqref{cond:regional-lya-bb} via Satisfiability-Modulo-Theory (SMT) queries
    \item Construction of counterexample samples
\end{compactitem}

\paragraph{Construction of Regions}
Our insight is to use samples as vertices of regions, i.e., a triangulation of sampled states,
so it is easy to obtain samples and the regions at the same time.
We choose Delaunay triangulation which generates a cover composed of \emph{simplices}.
Formally, given a set of samples $S=\{(\xj_j,\yj_j)\}_j$ with at least $n+1$ samples,
we can construct a Delaunay triangulation $\Cover=\{\Ri\}$ from $X=\{\xj_j \mid (\xj_j,\yj_j) \in S\}$,
and each region
\(
\Ri = \{\sum_{j=0}^n \lambda_j \xj_j \mid \bigwedge_{j=0}^n 0 \leq \lambda_j \leq 1 \land \sum_{j=0}^n \lambda_j = 1 \}
\)
is a simplex defined by $n+1$ affinely independent vertices $\xj_0,\dotsc,\xj_n \in X$.
Each region is the convex hull of its vertices by definition, which is recommended by \cref{thm:provable-convex}.
For any two distinct regions $\Reg_i\neq\Reg_{i'} \in \Cover$, their interiors,  $\INT(\Reg_i)$ and $\INT(\Reg_{i'})$, are disjoint.
Hence, these regions \emph{overlap only on facets} as recommended by \cref{remark:cover}.
Additionally, the Delaunay triangulation method supports \emph{incrementally adding more samples} to generate finer triangulations.
An example triangulation of $\States$ is shown in \cref{fig:vdp-cover} in \cref{sec:eval}.

\paragraph{Verification with SMT Queries}
With a triangulation $\Cover$, we use an SMT solver to falsify Condition~\eqref{cond:regional-lya-bb} for each simplex $\Reg$ defined by $S = \{(\xj_j, \yj_j)\}_{j=0\dotsc n}$.
The SMT query to falsify Condition~\eqref{cond:regional-lya-bb} is to search for a state $\vx \in \Reg$ by seeking $\lambda_0,\dotsc,\lambda_n$ that satisfies the formula below:
\begin{small}
\begin{multline*}
\exists \lambda_0 \in \Reals,\dotsc,\lambda_n \in \Reals, \bigwedge\nolimits_{j=0}^n 0 \leq \lambda_j \leq 1 \land \sum\nolimits_{j=0}^n \lambda_j = 1\ \land \\
\text{let } \vx = \sum\nolimits_{j=0}^n \lambda_j \xj_j,\quad
\vx \in \States \land \bigwedge\nolimits_{j=0}^n \norm{\derlya}\lipreg\norm{\vx-\xj_j} + \derlya \cdot \yj_j \geq 0
\end{multline*}
\end{small}%
The complexity of each SMT query depends on $\lya$ and $\gradlya$.
If $\lya$ is a polynomial of $\vx$, then the SMT query is decidable but NP-hard as discussed in Appendix~\ref{appx:equi-sat}.

\paragraph{Counterexample Construction}
If the SMT solver returns one or more satisfiable assignments for $\lambda_0,\dotsc,\lambda_n$,
then $\lya$ violates Condition~\eqref{cond:regional-lya-bb}.
We compute the counterexample state(s) $\vx_c = \sum_{j=0}^n \lambda_j \xj_j$,
obtain $\vy_c = \fbb(\vx_c)$, and add $(\vx_c, \vy_c)$ to the set of samples $S$.
To distinguish between Cases~\ref{case:true} and~\ref{case:spur},
we simply evaluate $\gradlya(\vx_c)\cdot\vy_c$.
If $\gradlya(\vx_c)\cdot\vy_c \geq 0$, then we know $\lya$ is falsified~(Case~\ref{case:true}).
Otherwise~(Case~\ref{case:spur}),
the incremental triangulation is constructed from new samples.
We also consider an inconclusive SMT query, e.g., due to time limits.
Our design, which preserves soundness, selects $\vx_c = \sum_{j=0}^n \xj_j/(n+1)$ to refine the cover.

\paragraph{Neither $\dist$-provable nor Falsified Candidates}
For a user-specified $\dist$, we can stop the verifier when we cannot falsify $\lya$ but still find a counterexample in a small enough simplex $\Reg$, i.e., \(\mathrm{diam}(\Reg) \leq \dist/2\cdot\sqrt{2(n+1)/n}\).
According to \cref{thm:provable-convex},
we have shown that $\lya$ is not $\dist$-provable,
and the user may lower the $\dist$ value.


\section{Learning in Convex Sets}\label{sec:learner}
In this section, we aim to design a learner for our \CEGIS algorithm ensuring the termination.
We achieve this by reducing the Lyapunov function synthesis problem to the \emph{convex feasibility problem}~(\cref{def:convex-feasibility}).
This will allow us to learn a Lyapunov function by cutting-plane methods in \cref{sec:cegus}.

We first provide our design choice on the hypothesis space for the learner,
i.e., the parameterized function template $\lya_\vpara$.
We describe our criteria on the function template and show a nontrivial example template satisfying our criteria.
We then provide an analytic center-based learner using samples and convex optimization.
The learner not only proposes a candidate compatible with the given set of samples
but also reports when no compatible candidate exists.

\subsection{Convex Set of Solutions in Hypothesis Space}\label{subsec:convex-hyp}

We first slightly abuse the notations and define two parameterized function templates $g_\vx: \Reals^d \to \Reals$ and $h_\vx: \Reals^d \to \Reals$ where $\vpara$ becomes the input to the functions.
That is,
\(g_\vx(\vpara) \coloneqq -\lya_\vpara(\vx)\) and
\(h_\vx(\vpara) \coloneqq \gradlya_\vpara(\vx) \cdot \fbb(\vx)\).
Our design choice is to enforce that both $g_\vx$ and $h_\vx$ are \emph{convex functions with respect to $\vpara$} for any state $\vx \in \States$.
This design choice ensures that the set of true Lyapunov functions in the hypothesis space is a convex set.
As a result, this allows us to reduce the Lyapunov synthesis problem to a convex feasibility problem.

More precisely, we define two sets,
$\HypsPD \coloneqq \{\vpara \mid \forall \vx \in \States. \lya_\vpara(\vx) > 0\}$ and 
$\HypsSTB \coloneqq \{\vpara \mid \forall \vx \in \States. \gradlya_\vpara(\vx) \cdot \fbb(\vx) < 0\}$.
By definition, $\HypsPD$ encodes all positive definite functions (Condition~\eqref{cond:lya-pd}) in the hypothesis space,
and $\HypsSTB$ encodes the set of all functions satisfying Condition~\eqref{cond:lya-stab}.
The Lyapunov synthesis problem is thus to search for $\vpara \in \HypsPD \cap \HypsSTB$.
Now we know $\HypsPD = \bigcap_{\vx\in\States} \{\vpara \mid g_\vx(\vpara) < 0\}$ by definition.
Because $g_\vx$ is convex, $\{\vpara \mid g_\vx(\vpara) < 0\}$ for every $\vx\in\States$ is a convex set, so their intersection $\HypsPD$ is also a convex set.
Similarly, because $h_\vx$ is convex, $\HypsSTB=\bigcap_{\vx\in\States} \{\vpara \mid h_\vx(\vpara) < 0\}$ is also a convex set.
Hence, the Lyapunov synthesis problem is a convex feasibility problem under this design.

To avoid analyzing the convexity of $h_\vx$ which may require a closed-form expression of $\fbb$,
we can use a stronger requirement that $g_\vx$ is \emph{linear with respect to $\vpara$}.
Because $\lya_\vpara(\vx) = -g_\vx(\vpara)$,
$\lya_\vpara$ is linear with respect to $\vpara$ as well.
For any two parameters $\vpara_1, \vpara_2 \in\Reals^d$ and two scalars $\alpha_1, \alpha_2 \in\Reals$,
\begin{align*}
     h_{\vx}(\alpha_1\vpara_1 + \alpha_2\vpara_2)
    &= \gradlya_{\alpha_1\vpara_1 + \alpha_2\vpara_2}(\vx) \cdot \fbb(\vx) \\
    &= [\frac{\partial \lya_{\alpha_1\vpara_1 + \alpha_2\vpara_2}(\vx)}{\partial x_1}\,
        \dotsc
        \frac{\partial \lya_{\alpha_1\vpara_1 + \alpha_2\vpara_2}(\vx)}{\partial x_n}
       ] \cdot \fbb(\vx) \\
    &= [\frac{\partial \alpha_1\lya_{\vpara_1}(\vx) + \alpha_2\lya_{\vpara_2}(\vx)}{\partial x_1}\,
        \dotsc
        \frac{\partial \alpha_1\lya_{\vpara_1}(\vx) + \alpha_2\lya_{\vpara_2}(\vx)}{\partial x_n}
       ] \cdot \fbb(\vx) \\
    &= \alpha_1\gradlya_{\vpara_1}(\vx)\cdot \fbb(\vx) + \alpha_2\gradlya_{\vpara_2}(\vx)\cdot \fbb(\vx) \\
    &= \alpha_1 h_{\vx}(\vpara_1) + \alpha_2 h_{\vx}(\vpara_2)
\end{align*}
Then, $\gradlya_\vpara$ is linear with respect to $\vpara$ by definition.
Therefore, $h_\vx$ is also linear with respect to $\vpara$.
It may seem very restrictive to require the linearity with respect to $\vpara$.
Here, we give a concrete nontrivial template with linearity.

\newcommand{\Tanh}{\ensuremath{\mathrm{Tanh}}}
\begin{example}[Templates with Transcendental Functions]\label{ex:quad-lya}
Let $\Tanh(\vx)$ denote applying $\tanh$ on each element $x_i$ in $\vx$.
Consider learning from a template function, 
\(
\lya_\vpara(\vx) = \Tanh(\vx)^T \Theta \Tanh(\vx)
\),
where $\vpara \in \Reals^{n^2}$ is the flatten vector of the $n\times n$ matrix $\Theta$.
It is easy to show that $g_\vx(\vpara)=-\lya_\vpara(\vx)$ is, in fact, linear with respect to the flattened vector $\vpara$.
Formally, given a state $\vx \in \States$, 
consider any two parameters $\vpara_1$ and $\vpara_2$ and two scalars $\alpha_1, \alpha_2 \in\Reals$:
\begin{small}
\[
g_{\vx}(\alpha_1\vpara_1 + \alpha_2\vpara_2)
    = -\Tanh(\vx)^T (\alpha_1 \Theta_1 + \alpha_2 \Theta_2) \Tanh(\vx)
    = \alpha_1 g_{\vx}(\vpara_1) + \alpha_2 g_{\vx}(\vpara_2)
\]
\end{small}%
We see that the non-convex function $\Tanh$ does not affect the convexity over $\vpara$.
Further, we can easily compute the gradient thanks to the linearity.
By expanding the matrix multiplication and the function $\Tanh$,
\begin{align*}
g_{\vx}(\vpara)
    &= -\Tanh(\vx)^T \Theta \Tanh(\vx) \\
    &= -\sum_{i=1}^n \sum_{j=1}^n \tanh(x_i){\cdot}\tanh(x_j){\cdot}\theta_{(i-1){\cdot}n + j}
\end{align*}
Hence, for $g_\vx(\vpara)$ with $\vx=[x_1\dotsc x_n]^T$,
we can compute the (sub)-gradient $\va = [a_1\dotsc a_{n^2}]^T \in \Reals^{n^2}$ as below:
\[
a_{(i-1){\cdot}n + j} = -\tanh(x_i){\cdot}\tanh(x_j), \quad \text{ for } i=1...n, j=1...n
\]
Templates in other existing works,
e.g., linear combinations of monomials in~\cite{ravanbakhsh_learning_2019}
and semidefinite matrices in~\cite{chen_learning_2021-1},
are also linear with respect to the parameters.
\end{example}

\subsection{Analytic-Center--Based Learner}\label{subsec:learn}

Now, we design an analytic center-based learner for \CEGIS.
Without loss of generality, we assume the initial parameter set $\Hyps_0 \subseteq \Reals^d$ is a unit hypercube center at $\zero$.
We represent
\(
\Hyps_0 = \bigcap_{i=1}^d \{\vpara \mid \ve_i\cdot\vpara < -\frac{1}{2}\} \cap \{\vpara \mid \ve_i\cdot\vpara < \frac{1}{2}\}
\)
as the intersection of halfspaces
where $\ve_i$ is the $i$-th basis vector of $\Reals^d$.
Recall the observation compatibility in \cref{def:compatible}.
Given the samples $S=\{(\xj_i, \yj_i)\}_i$ with its size $|S|$,
we define $\Hyps_S$ as the set of all parameters compatible with $S$,
i.e.,
\(
\Hyps_S 
    = \Hyps_0 \cap \bigcap_{i=1}^{|S|} \{\vpara \mid \lya_\vpara (\xj_i) > 0 \land \gradlya_\vpara(\xj_i) \cdot \yj_i < 0\}
\).
$\Hyps_S$ is a polytope formed by an intersection of halfspaces
using the definitions in \cref{subsec:convex-hyp},
so checking the strict feasibility of $\Hyps_S$ is therefore solvable with linear programming.
If $\Hyps_S$ is not strictly feasible,
we know no candidate is compatible with the samples $S$.

Otherwise,
the learner computes the analytic center of $\Hyps_S$ by solving the following convex optimization~\cite{goffin_complexity_1996}:
\begin{small}
\begin{equation*}
\cand = \arg\max_{\vpara\in\Reals^d}
       \sum\limits_{i=1}^d \ln{(\frac{1}{2} + \ve_i \cdot \vpara)} + \ln{(\frac{1}{2} - \ve_i \cdot \vpara)} +
       \sum\limits_{i=1}^{|S|} \ln(\lya_\vpara (\xj_i)) + \ln(-\gradlya_\vpara(\xj_i) \cdot \yj_i))
\end{equation*}
\end{small}%
The learner then proposes $\lya_{\cand}$ as the candidate function.
We are now ready to describe and analyze our \CEGIS algorithm in \cref{sec:cegus}.


\section{Black-Box \CEGIS}\label{sec:cegus}

\begin{figure}[t]
    \centering
    \scalebox{0.9}
    {\tikzstyle{process} = [rectangle, minimum width=15mm, minimum height=8mm, draw=black]
\tikzstyle{verify} = [rectangle, text width=33mm, minimum height=13mm, draw=black, fill=white]
\tikzstyle{decision} = [diamond, minimum width=15mm, minimum height=0.5cm, aspect=3, text centered, draw=black]
\tikzstyle{sum} = [circle, minimum width=0.25cm, minimum height=0.25cm, text centered, draw=black]
\tikzstyle{arrow} = [thick,->,>=stealth]

\pgfdeclarelayer{bg}
\pgfsetlayers{bg,main}

\begin{tikzpicture}[node distance=15mm,scale=1.0,every node/.style={transform shape}]
\node (init) {};
\node (learn) [process, text width=38mm, right=13mm of init] {%
    \textbf{Learner}\\ Propose $\lya_\vpara$ compatible with $S$ by analytic center};

\node (verify2) [verify, below=8mm of learn, xshift=-1mm, yshift=2mm] {};
\node (verify1) [verify, below=8mm of learn, xshift=0mm, yshift=0mm] {};
\node (verify0) [verify, below=8mm of learn, xshift=1mm, yshift=-2mm] {%
    \textbf{Regional Verifier}\\ Check Condition~\eqref{cond:regional-lya-bb} on $\lya_\vpara$ with $\Ri$ and $S_i$};

\node (cover) [process, text width=20mm, above right=-2mm and 8mm of verify1] {%
    \textbf{Covering}\\ Build cover $\Cover$};

\node (union) [sum, below right=-5mm and 19mm of verify1] {$\cup$};

\node (add)[process, text width=25mm, right=38mm of learn] {%
    Add samples \\
    $S \gets S \cup S_c$};
\node (sample)  [process, text width=27mm] at (add |- union) {%
    Obtain $\yj=\fbb(\xj)$ for each $\xj \in \bigcup X_i$};

\draw [arrow] (init) -- node [above,text width=13mm] {Initial \\ $S\gets S_0$} (learn);
\draw [arrow] (learn) -- node [left, align=right] {Candidate $\lya_\vpara$} (verify1);

\begin{pgfonlayer}{bg}

\draw [arrow] (cover) -- (verify2);
\draw [arrow] (cover) -- node [below right] {$\Ri \in \Cover, S_i \subseteq S$} (verify1);
\draw [arrow] (cover) -- (verify0);

\draw [draw=none] (cover) -- node [rectangle, dashed, draw=black,
    minimum width=7cm, minimum height=23mm, xshift=-24mm, yshift=0.5mm,
    label={[anchor=north]\large\textbf{Verifier}}] {} (union);

\draw [arrow] (verify2) -- (union);
\draw [arrow] (verify1) -- node[below] {$X_i$} (union);
\draw [arrow] (verify0) -- (union);
\end{pgfonlayer}

\draw [arrow] (union) -- (sample);
\draw [arrow] (sample) -- node[right, xshift=-8mm] {New Samples $S_c$} (add);
\draw [arrow] (add) -- node[above] {$S$} (learn);
\draw [arrow] (add) -- node[above left] {$S$} (cover);

\end{tikzpicture}}
    \caption{Architecture for black-box \CEGIS of Lyapunov functions. The detailed decision flow is in \cref{alg:ceglya}.}
    \label{fig:bbcegus-flow}
\end{figure}

\newcommand{\FBBCEGuS}{\ensuremath{\mathsf{BBCEGuS}}\xspace}
\newcommand{\FLearnV}{\ensuremath{\mathsf{Learner}}\xspace}
\newcommand{\FGenCover}{\ensuremath{\mathsf{Covering}}\xspace}
\newcommand{\FEstLip}{\ensuremath{\mathsf{EstimateLipschitz}}\xspace}
\newcommand{\FCexInRj}{\ensuremath{\mathsf{RegionalVerifier}}\xspace}

In this section, we explain our general \CEGIS architecture to find a Lyapunov function using our verifier in \cref{subsec:cegar} and our learner in \cref{subsec:learn}.
\cref{fig:bbcegus-flow} sketches the three major components and the data exchange in our architecture:
\begin{compactitem}
    \item \FLearnV: propose candidates using analytic centers,
    \item \FGenCover: generate a cover of the region of interest $\States$,
    \item \FCexInRj: verify each regional verification condition in parallel.
\end{compactitem}
Compared with \cref{fig:wbcegus-flow},
only the verifier is modified to build a cover of the state space for checking our regional verification conditions for each region.
We first discuss \cref{alg:ceglya}, an iterative \CEGIS algorithm, in \cref{subsec:iter-cegus}.
We then provide a termination guarantee of our \CEGIS algorithm based on the termination guarantee of \ACCPM in \cref{subsec:accpm}.
Lastly, we describe our implementation and techniques to speed up our \CEGIS algorithm in \cref{subsec:impl}.

\subsection{Iterative \CEGIS Algorithm for Black-Box Systems}\label{subsec:iter-cegus}

\begin{algorithm}[t]
\caption{Iterative \CEGIS for black-box systems.}\label{alg:ceglya}

\begin{algorithmic}[1]
\State $\fbb, \lipbb, \States\subseteq \Dom \setminus \{\zero\}$, $S_0$, $\dist$ and $\rob$ are given. $\lipreg$ for some regions $\Reg$ are optional.
\State $\mathsf{diam\_thres} \gets \frac{\dist}{2}\sqrt{\frac{2(n+1)}{n}}$ \Comment{Diameter threshold from $\dist$}
\State $\mathsf{max\_k} \gets \min\limits_{k\in\mathbb{N}} k \text{ subject to } \frac{\rob^2}{d} \geq  \frac{\frac{1}{2} +2d\ln(1+\frac{k+1}{8d^2})}{2d+k+1}$ \Comment{Iteration limit from $\rob$}
\State $S \gets S_0;\quad \Cover \gets \{\Dom\}$ \Comment{Initial dataset and a trivial cover}

\For{$k \gets 1\dotsc\mathsf{max\_k}$} \label{line:loop-head}
    \State $S_L \gets \{(\xj, \yj) \in S \mid \xj \in \States\}$ \Comment{Use only samples in \ROI for learning}
    \State $\lya_\vpara \gets \FLearnV(S_L)$ \label{line:learn}
    \If{$\lya_\vpara = \bot$}
        \Return ``No Lyapunov functions" \label{line:no-lya}
    \EndIf
    \State $\mathsf{stop\_refine} \gets \mathsf{false}$
    \Repeat \Comment{Refinement loop to verify $\lya_\vpara$}
        \If{$\mathsf{stop\_refine}$}
            \Return ``$\lya_\vpara$ is neither $\dist$-provable nor falsified." \label{line:no-cover}
        \EndIf

        \State $\Cover \gets \FGenCover(\Cover, \States, S)$  \Comment{Update cover with samples}
        \ForAll{$\Ri \in \Cover$ \label{line:inner}}
            \State $S_i \gets \{ (\xj, \yj) \in S \mid \xj \in \Ri\}$
            \State $\lipbb_i \gets (\exists \Reg. \Reg_i\subseteq \Reg)?\; \lipreg: \lipbb$ \Comment{Pick smaller Lipschitz bounds}
            \State $X_i \gets \FCexInRj(\lya_\vpara, \States, \Ri, S_i, \lipbb_i)$
            \If{$(X_i \neq \emptyset \land \mathrm{diam}(\Reg_i) \leq \mathsf{diam\_thres}$}
                \State $\mathsf{stop\_refine} \gets \mathsf{true}$
            \EndIf
        \EndFor
        \If{$\bigcup X_i = \emptyset$}
            \Return $\lya_\vpara$  \Comment{Found a Lyapunov function} \label{line:found}
        \EndIf
        \State $S_c \gets \{(\xj, \fbb(\xj)) \mid \xj \in \bigcup X_i\}; \quad S \gets S \cup S_c$ \Comment{Get new samples}
    \Until{$\exists (\xj, \yj)\in S_c.\ \xj\in\States \land \lya_\vpara(\xj) \leq 0 \land \gradlya_\vpara(\xj)\cdot \yj \geq 0$}\label{line:falsify}
    \State \Comment{$\lya_\vpara$ is falsified. Continue to learn a new candidate}
\EndFor
\State \Return ``No $\rob$-robust $\dist$-provable Lyapunov functions."\label{line:terminate}
\end{algorithmic}

\end{algorithm}

\cref{alg:ceglya} shows the pseudocode of an iterative implementation.
It takes as input
an executable function $\fbb:\Reals^n\mapsto\Reals^n$,
a Lipschitz bound $\lipbb$ in domain $\Dom$,
and the \ROI $\States \subseteq \Dom \setminus \{\zero\}$.
It can be configured with
an initial set of samples $S_0$,
a threshold $\dist$ for $\dist$-provability to derive the diameter threshold $\mathsf{diam\_thres}$,
a robustness parameter $\rob$ to derive the iteration limit $\mathsf{max\_k}$,
and and optionally regional Lipschitz bounds $\lipreg$ for some regions $\Reg$.

Overall, \cref{alg:ceglya} seeks for both a Lyapunov candidate $\lya_\vpara$ and a cover $\Cover$, and it repeatedly learns $\lya_\vpara$ and updates $\Cover$ using counterexamples until they pass checks by \FCexInRj for all regions.
It starts with an initial dataset $S=S_0$ and uses $S$ to learn a new Lyapunov candidate $\lya_\vpara$.
If the learner can learn a new candidate~($\lya_\vpara\neq\bot$),
it verifies this Lyapunov candidate $\lya_\vpara$ by refining the current cover $\Cover$ and checking each region $\Ri$,
and it obtains a possibly empty set of counterexamples $\bigcup X_i$.
It then updates the dataset and repeats until any new data falsifies the current candidate (\cref{line:falsify}).
It will then leave the repeat-loop and learn a new candidate.
\cref{alg:ceglya} terminates in the following situations:
\begin{compactitem}
    \item \FLearnV{} can't find a candidate in its hypothesis space (\cref{line:no-lya}).
    \item There is no counterexample, so we found a Lyapunov function $\lya_\vpara$ (\cref{line:found}).
    \item The diameter of any region that needs refinement is too small (\cref{line:no-cover}).
    \item It reaches the limit of iterations and hence no $\rob$-robust candidates (\cref{line:terminate}).
\end{compactitem}
Therefore, \cref{alg:ceglya} either synthesizes a Lyapunov function
or concludes the absence of a $\rob$-robust $\dist$-provable candidate on termination.

\subsection{Termination by \ACCPM}\label{subsec:accpm}

Recall from \cref{subsec:convex-hyp}: Both $g_\vx$ and $h_\vx$ are linear with respect to $\vpara$ by design,
and finding a Lyapunov function is to find $\vpara \in \HypsPD \cap \HypsSTB$,
i.e., a convex feasibility problem.
Our goal is to show that \cref{alg:ceglya} solves the convex feasibility by \ACCPM~\cite{atkinson_cutting_1995},
which always terminates according to \cref{thm:accpm}.
We show that the verifier serves as a separating oracle,
and the learner from \cref{subsec:learn} finds the analytic center of the polytope defined by the separating hyperplanes.

Assuming at the $k$-th iteration that the learner proposes $\lya_{\cand_k}$ at \cref{line:learn},
so $\cand_k$ is the analytic center of $\Hyps_{S_k}$.
The verifier then falsifies $\lya_{\cand_k}$ at \cref{line:falsify}.
A counterexample $\xj$ and $\yj=\fbb(\xj)$ shows that $\lya_{\cand_k}$ violates either Condition~\eqref{cond:lya-pd} or~\eqref{cond:lya-stab},
i.e., $\lya_{\cand_k}(\xj) \leq 0$ or $\gradlya_{\cand_k}(\xj)\cdot \yj \geq 0$,
which implies $\cand_k \notin \HypsPD \cap \HypsSTB$.
The separating hyperplane between $\cand_k$ and $\HypsPD \cap \HypsSTB$ can be constructed as follows.
If $\cand_k \notin \HypsPD$, we know $g_{\xj}(\cand_k) \geq 0$ because $\lya_{\cand_k}(\xj) \leq 0$.
Since $g_{\xj}$ is linear,
we find a separating hyperplane $(\nabla g_{\xj}(\cand_k), 0)$ according to Proposition~\ref{prop:grad-plane}.
Similarly, If $\cand_k \notin \HypsSTB$,
we find a separating hyperplane $(\nabla h_{\xj}(\cand_k), 0)$.
Note that $(\nabla g_{\xj}(\cand_k), 0)$ or $(\nabla h_{\xj}(\cand_k), 0)$ is only for proof and never explicitly constructed.

We now consider $\lya_{\cand_{k+1}}$ learned from the samples $S_{k+1}$ at the $(k+1)$-th iteration.
For simplicity, we consider that  only one pair of a counterexample state $\xj$ and its output $\yj$ is added to $S$, i.e., $S_{k+1} = S_k \cup \{(\xj, \yj)\}$.
By design,
\begin{small}
\begin{align*}
\Hyps_{S_{k+1}}
    &= \Hyps_{S_k} \cap \{\vpara\mid g_{\xj}(\vpara) < 0\} \cap \{\vpara\mid h_{\xj}(\vpara) < 0\} \\
    &= \Hyps_{S_k} \cap \{\vpara\mid \nabla g_{\xj}(\cand_k) \cdot\vpara < 0\}
                   \cap \{\vpara\mid \nabla h_{\xj}(\cand_k) \cdot\vpara < 0\}
\end{align*}
\end{small}%
Therefore, the new polytope $\Hyps_{S_{k+1}}$ excludes $\cand_k$,  the analytic center of $\Hyps_{S_k}$, using the hyperplanes $(\nabla g_{\xj}(\cand_k), 0)$ or $(\nabla h_{\xj}(\cand_k), 0)$,
and \cref{alg:ceglya} actually implements \ACCPM.
Lastly, we know the solution set $\HypsPD \cap \HypsSTB$ does not contain a $\rob$-ball after $\mathsf{max\_k}$ iterations by \cref{thm:accpm}.
This is similar to Lyapunov candidates with robust compatibility defined in~\cite[Definition 7]{ravanbakhsh_learning_2019}.

\subsection{Implementation and Speed Up}\label{subsec:impl}

We implemented a prototype of our \CEGIS algorithm in Python and released it on GitHub at \url{https://github.com/CyPhAi-Project/pricely} as open-source software.
We use existing convex optimization libraries \CVXPY~\cite{diamond_cvxpy_2016} for the learner and the \dReal SMT solver~\cite{gao_dreal_2013} for the verifier.
For the hypothesis space, we use the template $\lya_\vpara(\vx) = \frac{1}{2} \vx^T \Theta \vx$
where $\Theta$ is a \emph {symmetric} $n\times n$ matrices constructed from the parameter $\vpara\in \Reals^d$ with $d = n(n+1)/2$.
Additional details on the convex optimization for finding analytic centers are available in \cref{appx:learner}.

We implement two simple techniques to speed up \cref{alg:ceglya}, parallelizing and caching SMT queries.
Observe the for-loop at \cref{line:inner},
the regional verification for each individual region does not depend on the results of other regions.
Therefore, we can parallelize the loop and solve SMT queries in parallel.
We further cache the SMT query for a region where no counterexamples are found, i.e. UNSAT.
If the region is not modified after updating the cover,
we immediately know that the same SMT query is UNSAT and avoid solving the query again.


\newcommand{\lb}{0.1}
\newcommand{\ub}{\ensuremath{r}}
\newcommand{\YES}{\checkmark}
\newcommand{\NO}{}
\newcommand{\INC}{{$\triangle$}}
\newcommand{\UNK}{{$-$}}
\newcommand{\TO}{\textcolor{red}{>4 hrs}}
\newcommand{\SL}{\textcolor{red}{>500k}}
\newcommand{\mrow}[2]{\multirow{#1}{*}{{#2}}}
\newcommand{\mcol}[2]{\multicolumn{#1}{|c|}{{#2}}}
\newcommand{\others}{\textup{\texttt{Trans}}\xspace}
\newcommand{\fossil}{\textup{\texttt{Polys}}\xspace}

\section{Evaluation}\label{sec:eval}

For evaluation, we use two groups of benchmarks, namely \others and \fossil,
to evaluate the performance of our black-box \CEGIS approach.
For \others listed in \cref{tab:exp-others}, we study 4 benchmarks that are \emph{locally stable},
and $\fbb$ in each benchmark may consist of \emph{transcendental functions}.
We use all 3 benchmarks in~\cite{zhou_neural_2022}, namely Van der Pol, unicycle path following, and inverted pendulum.
The 4\textsuperscript{th} benchmark is the Stanley path following controller~\cite{hoffmann_autonomous_2007}.
The vehicle dynamics for~\cite{hoffmann_autonomous_2007} contains sine and cosine,
and the controller is \emph{piecewise continuous} for input with saturation.
We note that the unicycle path following and inverted pendulum from~\cite{zhou_neural_2022} were intended for control synthesis.
Our setup instead certifies \emph{the system with the neural network controller} synthesized by~\cite{zhou_neural_2022}.

For \fossil listed in Table~\ref{tab:exp-fossil},
we aim to study how our approach is impacted by the size of the region of interest $\States$ using \emph{globally stable} systems with polynomial functions $\fbb$.
We pick 8 of 17 benchmarks from FOSSIL~\cite[Table~1]{abate_fossil_2021} and exclude the other 9 benchmarks that are not for Lyapunov synthesis or not continuous nonlinear dynamic.
To study how the \ROI $\States$ impacts our \CEGIS approach,
we follow the setup in~\cite{abate_formal_2021,abate_fossil_2021} to specify $\States = \{\vx \mid \lb \leq \norm{\vx} \leq \ub\}$ excluding a ball of radius \lb{} and varying $\ub$ between 1, 5, and 10.

For all benchmarks, we configure our tool with an initial set $S_0$ containing $6^n$ evenly spaced samples,
the $\dist$-provability threshold $\dist= 10^{-4}$,
i.e., $\mathsf{diam\_thres}=\frac{10^{-4}}{2}\sqrt{\frac{2(n+1)}{n}}$,
and the iteration limit $\mathsf{max\_k}=40$.
Details for each benchmark, such as the Lipschitz bounds $\lipbb$ and $\lipreg$ and the \ROI $\States$, are in Appendix~\ref{appx:benchmarks}.
All experiments are run on an Ubuntu workstation with 20 cores and 256GB RAM.

\subsection{Evaluation Result}

\begin{table}[t]
    \centering
    \caption{Performance on \others benchmarks from~\cite{zhou_neural_2022,hoffmann_autonomous_2007}.
    Time limit: 4 hours.
    Sample limit: 500K.
    Iteration limit: 40.}\label{tab:exp-others}
    \begin{tabular}{|l|c|r|r|r|r|r|}
        \hline
                 &      &   Time & \mcol{2}{Learn} & \mcol{2}{Verify} \\
        \cline{4-7}
        Name     &      &  (sec) & k &$|S_L|$& $|S|$ & $|\Cover|$ \\
        \hline
        Van der Pol
                 & \YES &   1.63 & 1 &    16 &   488 &    954 \\
        \hline
        Unicycle path
                 & \YES &  25.21 & 1 &    16 &  9825 &  19628 \\
        \hline
        Inverted pendulum
                 & \INC*&  54.55 & 1 &    16 & 78182 & 156342 \\
        \hline
        Stanley controller
                 & \YES &   7.79 & 1 &    36 &   829 &   1621 \\
        \hline
    \end{tabular}
\end{table}

Here we discuss the experiment results on the two groups of benchmarks.
\cref{tab:exp-others} provides the results for \others.
\cref{tab:exp-fossil} provides the results for \fossil with varying radii $\ub$ of $\States$.
In both tables, we report the \CEGIS outcome, the time usage, the number of \CEGIS iterations $k$,
the number of samples for learner $|S_L|$,
the number of samples for verifier $|S|$,
and the number of simplices in the cover $|\Cover|$.
The \CEGIS outcome is denoted as follows:
`\YES{}' means we synthesized a $\dist$-provable Lyapunov function.
`\INC{}' means we found a candidate that is neither $\dist$-provable nor falsified.
and `\UNK{}' means unknown outcome due to time or sample limits.
We further validate the last Lyapunov candidate for `\INC{}' or `\UNK{}' with a \dReal-based white-box verifier as in~\cite{chang_neural_2019},
and `*' means the last candidate is actually a Lyapunov function.

\paragraph{Results for \others}
Our prototype found Lyapunov functions for 3 out of 4 benchmarks.
For the inverted pendulum with a NN controller from~\cite{zhou_neural_2022},
we validate the candidate which is not $\dist$-provable to be a true Lyapunov function.
Our investigation shows that the NN controller reacts very rapidly to stabilize the pendulum;
hence, the Lipschitz bounds are at the order of $10^4$,
in contrast to the bounds of others at the order of 10 to $10^2$.
This requires a triangulation with simplices smaller than $\mathsf{diam\_thres}\approx8.7\cdot10^{-5}$ for verification;
hence our prototype stops.

\begin{figure}[t]
    \centering
    \includegraphics[height=35mm]{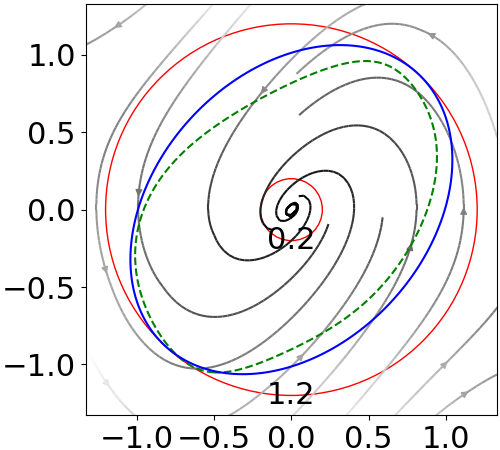}
    \hspace{1em}
    \includegraphics[height=35mm]{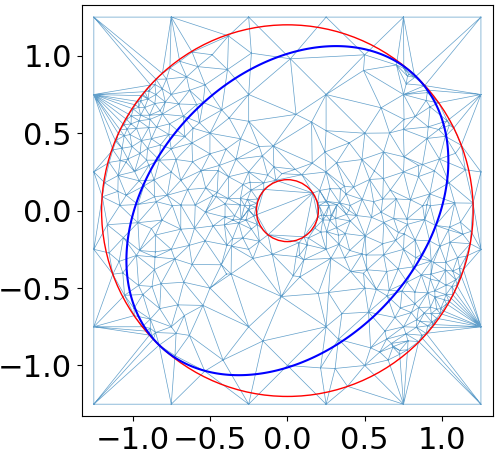}
    \caption{Comparison on Van der Pol. Phase portrait with BOAs~(\emph{Left}) and final triangulation covering $\States$~(\emph{Right}).
    The disk between the two \textcolor{red}{red} circles is $\States$. The \textcolor{blue}{blue} ellipse is our BOA,
    and the \textcolor{OliveGreen}{dashed green} contour is the BOA by~\cite{zhou_neural_2022}.}
    \label{fig:vdp-cover}
\end{figure}

We further studied the Van der Pol benchmark and roughly compared\footnote{
Due to the difference in the target problem, we do not compare the computation time because a fair comparison is difficult. See \cref{appx:subsec:compare-neurips2022} for the details.} our approach with~\cite{zhou_neural_2022}.
\cref{fig:vdp-cover}~(\emph{Left}) shows the basin of attractions~(BOA) in $\States$ from our prototype versus~\cite{zhou_neural_2022},
and both BOAs cover roughly the same neighborhood around the origin.
In~\cite{zhou_neural_2022},
the authors used $3000\times3000=9\cdot10^6$ evenly-spaced samples in the $(-1.5, 1.5)\times(-1.5, 1.5)$ box to ensure $\delta\leq5\cdot10^{-4}$.
In comparison, we use \emph{479 samples} and 936 simplices to find a $\dist$-provable Lyapunov function as seen in \cref{tab:exp-others}.
We calculated the diameters of simplices in the triangulation shown in \cref{fig:vdp-cover}~(\emph{Right}),
and the diameter ranges from 0.0225 to 0.578.
This suggests that
$\dist = 0.0225\cdot2/\sqrt{n/2(n+1)} \approx 0.0259 \;(n=2)$ is small enough by \cref{thm:provable-convex},
and Condition~\eqref{cond:regional-lya-bb} holds for plenty of larger simplices.
This showcases how our approach reduces the number from millions to hundreds of samples.

\begin{table}[t]
    \centering
    \caption{Performance on \fossil benchmarks from FOSSIL~\cite{abate_fossil_2021}.
    Time limit: 4 hours.
    Sample limit: 500k.
    Iteration limit: 40.}\label{tab:exp-fossil}
    \scriptsize
    \begin{tabular}{|l|r|c|r|r|r|r|r|}
        \hline
            Name &      &      &    Time  &\mcol{2}{Learn}&\mcol{2}{Verify}  \\
        \cline{5-8}
            $n$D &  \ub &      &    (sec) &  k &$|S_L|$&  $|S|$ & $|\Cover|$ \\
        \hline
        $\mathtt{nonpoly}_0$
                 &  $1$ & \YES &     0.80 &  1 &    16 &     53 &     84 \\
              2D &  $5$ & \YES &     2.79 &  2 &   245 &    620 &   1218 \\
                 & $10$ & \YES &     8.44 &  2 &    62 &   3579 &   7136 \\
        \hline
        $\mathtt{nonpoly}_1$
                 & $ 1$ & \YES &     0.78 &  1 &    16 &     96 &    170 \\
              2D & $ 5$ & \YES &  1520.35 &  6 & 11408 & 116843 & 233662 \\
                 & $10$ & \UNK*&       -- &  5 & 55338 &    \SL &     -- \\
        \hline
        $\mathtt{nonpoly}_2$
                 & $ 1$ & \YES &     7.81 &  1 &    56 &   1007 &   6092 \\
              3D & $ 5$ & \YES &   111.22 &  1 &    56 &   2776 &  17594 \\
                 & $10$ & \YES &  2819.90 &  5 & 25486 &  37675 & 243765 \\
        \hline
        $\mathtt{nonpoly}_3$
                 & $ 1$ & \YES &    48.07 &  1 &    56 &   2589 &  16353 \\
              3D & $ 5$ & \YES &   983.32 &  1 &    56 &  44415 & 280399 \\
                 & $10$ & \UNK*&      \TO & -- &    -- &     -- &     -- \\
         \hline
    \end{tabular}
    \hspace{1mm}
    \begin{tabular}{|l|r|c|r|r|r|r|r|}
        \hline
                   Name &      &      &    Time  & \mcol{2}{Learn}&\mcol{2}{Verify}       \\
        \cline{5-8}
            $n$D &  \ub &      &    (sec) &  k & $|S_L|$ & $|S|$ & $|\Cover|$ \\
        \hline
        $\mathtt{poly}_1$
                 & $ 1$ & \UNK*&       -- &  1 &      56 &   \SL &     -- \\
              3D & $ 5$ & \UNK*&       -- &  2 &     766 &   \SL &     -- \\
                 & $10$ & \UNK*&       -- &  2 &     742 &   \SL &     -- \\
        \hline
        $\mathtt{poly}_2$
                 & $ 1$ & \YES &    13.64 &  1 &    16 &    308 &    594 \\
              2D & $ 5$ & \YES &    22.76 &  7 &  1090 &   1183 &   2344 \\
                 & $10$ & \YES &   169.71 &  6 &  2583 &   3133 &   6244 \\
        \hline
        $\mathtt{poly}_3$
                 & $ 1$ & \YES &    25.43 &  1 &    16 &   2407 &   4792 \\
              2D & $ 5$ & \YES &    26.33 &  1 &    16 &   2664 &   5306 \\
                 & $10$ & \YES &   149.95 &  1 &    16 &   2762 &   5502 \\
        \hline
        $\mathtt{poly}_4$
                 & $ 1$ & \YES &    22.43 &  1 &    16 &   1111 &   2200 \\
              2D & $ 5$ & \YES &   619.95 &  2 &    62 &  92708 & 185394 \\
                 & $10$ & \UNK*&       -- &  3 &   604 &    \SL &     -- \\
        \hline
    \end{tabular}
\end{table}

\paragraph{Results for \fossil}
Our prototype successfully synthesized a Lyapunov function for 7 of 8 benchmarks when $\ub=1$ and $\ub=5$.
When $\ub=10$,
our prototype still succeeds for 4 benchmarks.
Moreover, even when our prototype terminated due to resource limits,
it still found candidates that are true Lyapunov functions.
As shown in \cref{tab:exp-fossil},
the 3D systems and a larger radius of $\States$ require a longer time.
This is because more samples $|S|$ are required for building a triangulation of $\States$.
This leads to more simplices in $\Cover$ and, thus, more SMT queries.
In comparison, the samples for learning $|S_L|$ are much fewer and stay the same for different $r$ for several benchmarks.
This is because a Lyapunov function is already found in the first iteration~($k=1$),
and $|S_L|$ does not increase since there will never be counterexamples for falsification.
In short, the analytic center-based learner is able to learn a good candidate with a few samples.
Our black-box verifier, however, may require a very fine triangulation with many SMT queries.


\section{Conclusion}\label{sec:conclusion}
In this paper, we presented a \CEGIS approach for certifying the Lyapunov stability of \emph{black-box} dynamical systems.
Our regional verification condition allows checking the Lie derivative of a Lyapunov function for a black-box system using counterexample-guided sampling.
We outline our design of the hypothesis space and the analytic center-based learner to efficiently synthesize a Lyapunov function,
and our \CEGIS algorithm guarantees the termination based on \ACCPM for solving convex feasibility.

Our evaluation showed that our approach is able to find a Lyapunov function with a few thousand samples for 2D and 3D systems for a small \ROI.
This is significantly fewer than the number of samples used in~\cite{zhou_neural_2022}.
The result also showed known scalability issues of black-box approaches:
The number of samples grows rapidly with respect to the system dimension and the size of \ROI.

The main assumption on known Lipschitz bounds can be addressed by integrating \emph{Lipschitz learning} methods~\cite{Wood1996EstimationOT,huang_sample_2023} into our \CEGIS flow.
Recent work~\cite{huang_sample_2023} not only estimates Lipschitz bounds for black-box functions but also provides theoretical bounds on the required number of samples.
Another assumption inherited from~\cite{zhou_neural_2022} is to collect samples arbitrarily.
In general,
\begin{inparaenum}[1)]
\item it can be costly to set the black-box system in arbitrary states for sampling, and
\item the system may never reach certain regions of states from normal initial conditions.
\end{inparaenum}
Synthesizing Lyapunov functions without this assumption will be an important future work.

Our work suggests quite a few extensions for certifying the stability of black-box systems.
Some obvious extensions are to handle more general systems, e.g., switched or hybrid systems,
to use different Lyapunov function templates,
e.g., piecewise affine functions~\cite{berger_counterexample-guided_2023}, rational polynomials~\cite{chesi_rational_2013}, etc.,
or to find the basin of attraction besides a Lyapunov function~\cite{chen_learning_2021}.
Controller synthesis, however, may not be a reasonable extension.
Due to the counterexample-guided nature, we suspect that such an approach will synthesize controllers that barely satisfy the stability.
Addressing the scalability issue for higher dimensional systems is another future direction.
Identifying and verifying only reachable regions of the state space may also help reduce the required number of samples~\cite{zhang_learning_2024}.
Recent advances in GPU computing motivate a brand-new design to accelerate \CEGIS with massive parallelization.


\subsubsection*{Acknowledgements.}
The authors would like to thank Prof. Kazumune Hashimoto from Osaka University and Dr. Aneel Tanwani from National Center of Scientific Research~(CNRS) in France for their valuable input on Lyapunov stability theory, existing works with SDP for searching Lyapunov candidates, and triangulation for partitioning.
This work is supported by the Japan Science and Technology Agency (Grant\# JPMJCR2012 and JPMJPR22CA) and Japan Society for the Promotion of Science (Grant\# 24K23861 and 22K17873).


%
%

\bibliographystyle{splncs04}
\bibliography{sample}

\appendix

\section{Acronyms and Symbols}\label{appx:symbols}
We provide acronyms in \cref{tab:acronyms}, symbols for describing the dynamical systems in \cref{tab:symbols}, and symbols for the hypothesis space for the learner in \cref{tab:sym-learning}.

\begin{table}[ht]
    \centering
    \caption{Acronyms}\label{tab:acronyms}
\begin{tabular}{ll}
     \ACCPM & Analytic Center Cutting-Plane Method \\
     BOA    & Basin of Attraction \\
     \CEGIS & Counter-Example Guided Inductive Synthesis \\
     ODE    & Ordinary Differential Equation \\
     NN     & Neural Network \\
     \ROI   & Region of Interest \\
     SMT    & Satisfiability Modulo Theory
\end{tabular}
\end{table}

\begin{table}[ht]
    \centering
    \caption{Symbols for Dynamical Systems}\label{tab:symbols}
    \begin{tabular}{r@{\hskip2pt}c@{\hskip2pt}l@{\hskip4pt}l}
    \hline
           $n$ & $\in$     & $\mathbb{N}$ & State dimensions \\
     $\cdot\;$ & $:$       & $\Reals^n \times \Reals^n \to \Reals$ & Inner product of two vectors \\
   $\norm{\ }$ & $:$       & $\Reals^n \to \Reals_{\geq 0}$ & Euclidean norm of a vector \\
        $\fbb$ & $:$       & $\Reals^n \to \Reals^n $ & Black-box right-hand side of ODE \\
        $\lya$ & $:$       & $\Reals^n \to \Reals$ & Differentiable Lyapunov candidate \\
    $\gradlya$ & $:$       & $\Reals^n \to \Reals^n$ & Gradient vector of the function $\lya$ \\
        $\Dom$ &$\subseteq$& $\Reals^n$ & Domain of states surrounding origin \\
     $\States$ &$\subseteq$& $\Dom \setminus \{\zero\}$ & Region of interest excluding origin\\
         $\Ri$ & $\subset$ & $\Dom$     & A region in the domain \\
$\lipbb, \lipreg$ & $\in$     & $\PosReals$ & Lipschitz bounds for $\fbb$ in  $\Dom$ or $\Reg \subseteq \Dom$\\
    $\vx, \xj$ & $\in$     & $\Reals^n$ & Any state $\vx$, a sampled state $\xj$ \\
$\yj=\fbb(\xj)$& $\in$       & $\Reals^n$ & The output for a sampled state $\xj$ \\
$\Ball_r(\vx)$ &$\subseteq$& $\Reals^n$  & A closed $n$-ball of radius $r$ around $\vx$ \\
    \hline
    \end{tabular}
\end{table}

\begin{table}[ht]
    \centering
    \caption{Symbols for Hypothesis Spaces}\label{tab:sym-learning}
    \begin{tabular}{r@{\hskip2pt}c@{\hskip2pt}l@{\hskip4pt}l}
    \hline
          $d$ & $\in$     & $\mathbb{N}$ & Parameter dimensions \\
     $\vpara$ & $\in$     & $\Reals^d$   & A parameter vector \\
    $\HypsPD$ &$\subseteq$& $\Reals^d$   & Positive definite candidates \\
   $\HypsSTB$ &$\subseteq$& $\Reals^d$   & Candidates decreasing along trajectories \\
    $\Hyps_S$ &$\subseteq$& $\Reals^d$   & Candidates compatible with samples in $S$\\
    \hline
    \end{tabular}
\end{table}

\section{Complete Proofs}\label{appx:proofs}

In this section, we provide complete proofs for the theorems shown in this paper.
We start with \cref{lem:lie-ub} which is an important intermediate result that is used to prove \cref{prop:bbox-regional} and \cref{thm:eps-equiv}, \ref{thm:grid-vc}, and \ref{thm:provable-convex}.

\begin{lemma}\label{lem:lie-ub}
For any two states $\vx, \vx_i \in \Dom$ and the sampled output $\vy_i = \fbb(\vx_i)$,
we have
\[
    |\liederlya - \derlya \cdot\vy_i| \leq \norm{\derlya}\lipbb\norm{\vx - \vx_i}
\]
Further, if $\vx \in \Ball_\dist(\vx_i)$, then we have
\[
   \norm{\derlya}\lipbb\norm{\vx - \vx_i} \leq M \lipbb \dist
\]
where $\lipbb$ is a Lipschitz bound for $\fbb$ in $\Dom$,
and $M \geq \sup\limits_{\vx\in \States}\norm{\derlya}$ is an upper bound on the norm of the gradient.
\end{lemma}
\begin{proof}
Here we apply Cauchy-Schwarz inequality\footnote{For norms other than the Euclidean norm, we can apply H\"{o}lder's inequality instead of Cauchy-Schwarz inequality.} and Lipschitz continuity.
\begin{align*}
    & |\liederlya - \derlya \cdot \vy_i| \\
   =& |\derlya \cdot (\fbb(\vx) - \fbb(\vx_i))| \\
\leq& \norm{\derlya}\norm{\fbb(\vx) - \fbb(\vx_i)} \qquad \text{(Cauchy-Schwarz)} \\
\leq& \norm{\derlya}\lipbb\norm{\vx - \vx_i} \qquad (\fbb \text{ is Lipschitz continuous})
\end{align*}
This proves the first inequality.
Further, $M \geq \norm{\derlya}$ by definition,
and $\norm{\vx - \vx_i}$ is bounded by $\dist$ because $\vx \in \Ball_\dist(\vx_i)$.
\begin{align*}
     & \norm{\vx - \vx_i} &\\
\leq & \norm{\vx - \vx_i} + \norm{\ctrl(\vx) - \vu_i} & \text{(Triangle Inequality)}\\
\leq & \norm{\vx - \vx_i} + \lipkk\norm{\vx - \vx_i}  & (\ctrl \text{ is Lipschitz cont.})\\
\leq & \dist                        & (\vx \in \Ball_\dist(\vx_i))
\end{align*}
Therefore,
\(
\norm{\derlya}\lipbb\norm{\vx - \vx_i} \leq M \lipbb (1+\lipkk) \dist
\)
\end{proof}

\subsection{Proof for \cref{thm:eps-equiv}}\label{appx:thm:eps-equiv}
\begin{proof}
First, we show that Condition~\eqref{cond:eps-cover} implies Condition~\eqref{cond:lya-stab}.
By simply applying Lemma~\ref{lem:lie-ub},
we have:
\begin{math}
\liederlya - \derlya \cdot \vy_i \leq M \lipbb \dist
\end{math}.
Adding both sides with Condition~\eqref{cond:eps-cover}, we derive
$\liederlya < -M \lipbb \dist < 0$.
The above holds for all $\dist$-balls around all samples $\vx_i$, which cover the entire $\States$.

We prove the other direction by contradiction.
We will show that Condition~\eqref{cond:lya-stab} contradicts the negation of Condition~\eqref{cond:eps-cover}.
Because $\States$ is compact and $\gradlya$ and $\fbb$ are continuous,
there exists $\beta > 0$ for Condition~\eqref{cond:lya-stab} so that 
\begin{equation}\label{cond:beta}
\forall \vx\in\States, \liederlya < -\beta   
\end{equation}
By assumption, there exists no $\dist$-cover to prove Condition~\eqref{cond:eps-cover}.
We arbitrarily choose a $\dist > 0$ satisfying
\(
{3M\lipbb} \dist \leq \beta
\),
and there exist two states $\vx_i \in \Dom$ and $\vx\in \Ball_\dist(\vx_i) \cap \States$ to falsify Condition~\eqref{cond:eps-cover}:
\begin{small}
\[
\begin{array}{rrl}
\derlya \cdot \vy_i \geq -2M\lipbb\dist
\Rightarrow &   \derlya \cdot \vy_i -M\lipbb\dist &\geq -3M\lipbb\dist \\
\Rightarrow &    \derlya \cdot\vy_i -M\lipbb\dist &\geq -\beta
\end{array}
\]
\end{small}%
We use $\derlya \cdot (\vy_i - \fbb(\vx)) \leq M\lipbb\dist$ from Lemma~\ref{lem:lie-ub}:
\begin{small}
\[
\derlya \cdot \vy_i - \derlya \cdot (\vy_i - \fbb(\vx)) \geq -\beta
\Rightarrow \derlya \cdot \fbb(\vx) \geq -\beta
\]
\end{small}%
This contradicts Condition~\eqref{cond:beta}.
By contradiction, Condition~\eqref{cond:lya-stab} implies Condition~\eqref{cond:eps-cover}.
\end{proof}
\begin{remark}
Note that Definition~\ref{def:provability} and Theorem~\ref{thm:eps-equiv} depend on the continuously differentiable $\lya$ so that the gradient $\gradlya$ is continuous in $\States$,
and hence there exists a bound $-\beta$ for the Lie derivative in the compact region $\States$.
\end{remark}

\subsection{Proof for \cref{prop:bbox-regional}}\label{appx:prop:bbox-regional}
\begin{proof}
Recall that we use $\derlya \cdot \yj$ to approximate $\liederlya$.
Let $\diff_{\xj}(\vx)$ denote \emph{the signed error} for an unobserved state $\vx$,
we first show the upper bound on $\diff_{\xj}(\vx)$ by Lemma~\ref{lem:lie-ub}:
\begin{align*}
\diff_{\xj}(\vx) &= \liederlya - \derlya \cdot \fbb(\xj) \\
           &\leq | \liederlya - \derlya \cdot \fbb(\xj) |
            \leq \norm{\derlya}\lipreg\norm{\vx - \xj}
\end{align*}
We then derive an upper bound of $\derlya \cdot \fbb(\vx)$ as below:
\begin{align*}
\liederlya 
     =\ & \diff_{\xj}(\vx) + \derlya \cdot \fbb(\xj) \\
  \leq\ & \norm{\derlya}\lipreg\norm{\vx - \xj} + \derlya \cdot \yj
        = \LieUB_{\xj, \yj}(\vx)
\end{align*}
Notice that each sample $(\xj,\yj)$ in $S$ leads to one upper bound for the state $\vx$,
and any upper bound below 0 is sufficient as described in Condition~\eqref{cond:regional-lya-bb}.
\end{proof}
A notable feature of the regional verification condition is that we may use multiple samples in $\Reg$.
We argue that the nearest sample $\xj$ of an unobserved state $\vx$ may not provide the tightest upper bound.
This is partly due to applying Cauchy-Schwarz inequality in Lemma~\ref{lem:lie-ub} to bound $\diff_\xj(\vx)$.
Especially when $\diff_\xj(\vx) < 0$, the upper bound can be loose even when $\norm{\vx - \xj}$ is small
because the norm is always non-negative.
Hence, checking multiple samples may actually prove Condition~\eqref{cond:regional-lya-bb} more easily.

\subsection{Proof for \cref{thm:grid-vc}}\label{appx:thm:grid-vc}
\begin{proof}
We first prove the ``if'' direction.
By \cref{prop:lya-regional} and~\ref{prop:bbox-regional}, we know that a regionally decreasing function $\lya$ for all regions $\Ri$ must be decreasing along all trajectories in $\States$.
Hence, $\lya$ is $\dist$-provably decreasing for some $\dist$ due to Theorem~\ref{thm:eps-equiv}.

We now prove the ``only if'' direction by proving that Condition~\eqref{cond:eps-cover} implies Condition~\eqref{cond:regional-lya-bb}.
Let $\{\vx_i\}_{i=1\dotsc N}$ be the $\dist$-cover for Condition~\eqref{cond:eps-cover},
we construct a cover $\Cover=\{\Reg_{i}\}$ by setting each region $\Reg_i = \Ball_\dist(\vx_i)$ and use the most conservative Lipschitz bound for $\Reg_i$, i.e., $\lipbb_{\Reg_i} = \lipbb$.
We choose a singleton set $S_i = \{(\vx_i, \vy_i)\}$ using the center $\vx_i \in \Reg_i$.
By Lemma~\ref{lem:lie-ub},
\(
-M\lipbb\dist \leq -\norm{\derlya} \lipbb \norm{\vx - \vx_i}
\).
By Condition~\eqref{cond:eps-cover}, we know for all $\vx \in \Ball_\dist(\vx_i) \cap \States$:
\begin{small}
\[
\begin{array}{rrl}
            \derlya \cdot \vy_i < -2M\lipbb\dist
\Rightarrow & \derlya \cdot \vy_i &< -2\norm{\derlya} \lipbb \norm{\vx - \vx_i} \\
\Rightarrow & \LieUB_{\vx_i, \vy_i}(\vx)  &< -\norm{\derlya} \lipbb \norm{\vx - \vx_i} \leq 0
\end{array}
\]
\end{small}%
This is exactly Condition~\eqref{cond:regional-lya-bb} if we use only one sample.
\end{proof}

\subsection{Distance Bound to the Nearest Vertex in a Convex Hull}
\begin{lemma}\label{lem:convex-hull-cover}
Let $\mathcal{V} = \{\vx_1, \vx_2, \dotsc, \vx_N\} \subset \Reals^n$ be a set of points and
let $\Reg=\mathrm{conv}(\mathcal{V})$ be the convex hull of $\mathcal{V}$.
W.o.l.g,
if $\Reg$ is covered by an $r$-ball centered at the origin, i.e.,
$\Reg \subseteq \Ball_r(\zero)$,
then $\Reg$ is also covered by the union of $r$-balls centered at points in $\mathcal{V}$, i.e.,
\(
\Reg \subseteq \bigcup_{\vx_i \in \mathcal{V}} \Ball_r(\vx_i)
\).
\end{lemma}
\begin{proof}
Based on \url{https://math.stackexchange.com/q/4203164},
we provide a proof in our notations.
By assumption, we know $\vx_i \in \Ball_r(\zero)$ for each point $\vx_i \in \mathcal{V}$.
For any point $\vx$ in the convex hull $\Reg$,
we can find $\vx = \sum_{i=1}^N \lambda_i \vx_i$ with $\lambda_i \geq 0$ and $\sum_{i=1}^N \lambda_i = 1$.
The squared distance from $\vx$ to any point $\vx_i \in \mathcal{V}$ is:
\begin{align*}
\norm{\vx - \vx_i}^2 = \norm{\vx}^2 - 2\vx \cdot \vx_i + \norm{\vx_i}^2
\end{align*}
The weighted mean squared distance using $\lambda_i$ as weights is:
\begin{align*}
\sum_{i=1}^N \lambda_i\norm{\vx - \vx_i}^2
   =& \norm{\vx}^2 - 2\vx\cdot\sum_{i=1}^N \lambda_i \vx_i + \sum_{i=1}^N \lambda_i \norm{\vx_i}^2 \\
   =& \norm{\vx}^2 - 2\norm{\vx}^2 + \sum_{i=1}^N \lambda_i \norm{\vx_i}^2 \\
\leq& -\norm{\vx}^2 + r^2 \qquad (\because \vx_i \in \Ball_r(\zero) \therefore \norm{\vx_i}^2 \leq r^2) \\
\leq& r^2
\end{align*}
In addition, let $\mathrm{nr}(\vx) = \mathrm{arg}\min_{\vx_i \in \mathcal{V}}\{\norm{\vx-\vx_i}\}$ returns the nearest point in $\mathcal{V}$ for any $\vx\in\Reg$.
Because all weights $\lambda_i\geq 0$ and $\sum_{i=1}^N \lambda_i = 1$,
we know
\[
\norm{\vx-\mathrm{nr}(\vx)}^2 = \sum_{i=1}^N \lambda_i\norm{\vx-\mathrm{nr}(\vx)}^2 \leq \sum_{i=1}^N \lambda_i\norm{\vx - \vx_i}^2 \leq r^2
\]
This implies that every $\vx \in \Reg$ must be covered by the $r$-ball around $\mathrm{nr}(\vx)\in\mathcal{V}$ that is nearest to $\vx$,
so $\Reg \subseteq \bigcup\limits_{\vx_i \in \mathcal{V}} \Ball_r(\vx_i)$.
\end{proof}

\subsection{Proof for \cref{thm:provable-convex}}\label{appx:thm:provable-convex}
\begin{proof}
First, because $\lya$ is $\dist$-provably decreasing,
we can find a sampled state $\vx_i \in\Dom$ for any state $\vx \in \Reg\cap\States$ so that 
$\norm{\vx_i - \vx}\leq \dist$ and
\(
\derlya \cdot \vy_i < -2M\lipbb\dist
\).

Second, by Jung's theorem~\cite[Theorem 2.6]{danzer_hellys_1963},
our requirement on the diameter of $\Reg$ ensures the existence of a ball enclosing $\Reg$ with radius $\frac{\dist}{2}$.
In combination with \cref{lem:convex-hull-cover}, the distance from any state $\vx \in \Reg$ to the nearest $\xj \in \mathcal{V}$ is bounded by $\frac{\dist}{2}$,
i.e., $\forall \vx\in\Reg, \min\limits_{\xj\in\mathcal{V}} \norm{\vx - \xj}\leq \frac{\dist}{2}$.
Hence, we know $\norm{\vx_i - \xj}\leq \norm{\vx_i - \vx} + \norm{\vx - \xj} \leq \frac{3}{2}\dist$.

We then derive the following from Condition~\eqref{cond:eps-cover}: 
\begin{small}
\[
                 \derlya \cdot \vy_i < -2M\lipbb\dist
\iff \derlya \cdot \vy_i + M\lipbb (\frac{3}{2}\dist + \frac{\dist}{2}) < 0
\]
\end{small}%
By \cref{lem:lie-ub},
$\derlya \cdot(\yj - \vy_i) \leq M \lipbb \frac{3}{2}\dist$.
Besides, $\norm{\derlya}\lipreg\norm{\vx - \xj} \leq M \lipbb \frac{\dist}{2}$.
We then derive for all $\vx\in \Reg\cap\States$:
\begin{small}
\[
\begin{aligned}
\Rightarrow && \derlya \cdot \vy_i + \derlya \cdot (\yj - \vy_i) + M\lipbb \frac{\dist}{2} & < 0 \\
\Rightarrow &&        \derlya \cdot \yj + \norm{\derlya}\lipreg\norm{\vx - \xj} & < 0
\end{aligned}
\]
\end{small}%
Therefore, $\lya$ is also regionally decreasing in $\Reg$ witness by $S$.
\end{proof}

\subsection{Equi-satisfiable SMT query for Regional Verification}\label{appx:equi-sat}
For a given region $\Reg$ and a set of samples $S=\{(\xj_j, \yj_j)\}_j$,
recall the original SMT query:
\begin{small}
\[
\exists \vx \in \Reg \cap \States,
\bigwedge_{(\xj_j,\yj_j)\in S} \norm{\derlya}\lipbb\norm{\vx-\xj_j} + \derlya \cdot \yj_j \geq 0
\]
\end{small}%

To remove the square root function in the norm,
we use a common technique of introducing auxiliary variables $\alpha$ and $r_j$ so that $0 \leq \alpha^2 \leq \norm{\derlya}^2$ and $0 \leq r_j^2 \leq \norm{\vx-\xj_j}^2$ for each $(\xj_j,\yj_j) \in S$.
We rewrite the above query as:
\begin{small}
\begin{multline*}
\exists \vx \in \Reg \cap \States, \exists \alpha\in\PosReals, \alpha^2 \leq \norm{\derlya}^2 \land\\
\bigwedge_{(\xj_j,\yj_j)\in S} \exists r_j \in\PosReals, r_j^2 \leq \norm{\vx-\xj_j}^2 \land \alpha\lipbb r_j + \derlya \cdot \yj_j \geq 0
\end{multline*}
\end{small}%
It is easy to show that the two queries are equi-satisfiable because the maximum values of $\alpha$ and $r_j$ is bounded by the respective norms,
and hence the existence of $\alpha$ and $r_j$ satisfying $\alpha\lipbb r_j + \derlya \cdot \yj_j \geq 0$ is \emph{equivalent} to the existence of $\vx$ satisfying $\norm{\derlya}\lipbb\norm{\vx-\xj_j} + \derlya \cdot \yj_j \geq 0$.
Further, the same satisfiable assignment of $\vx$ for one must satisfy the other SMT queries.

We assume the region $\Reg$ can be specified as a polytope by design.
The required background theory to solve the rewritten SMT query obviously depends on the gradient function $\derlya$.
If the Lyapunov candidate $\lya(\vx)$ is quadratic,
$\derlya$ is linear with respect to $\vx$.
The SMT query is a conjunction of quadratic constraints, which is equivalent to the feasibility of Quadratically Constrained Quadratic Programming~(QCQP) problem.
If the Lyapunov candidate $\lya(\vx)$ is a \emph{rational polynomial} of $\vx$~\cite{chesi_rational_2013},
we can simplify the denominator of $\derlya$,
and the SMT query is equivalent to the feasibility/emptiness of a \emph{basic semialgebraic set}.
The feasibility of a basic semialgebraic set is decidable.
It can be solved more efficiently if the set is convex, but it is NP-hard in the general case.

\section{Learn Quadratic Candidates}\label{appx:learner}
Recall the template $\lya_\vpara: \Reals^n \mapsto \Reals$ is as below:
\[
\lya_\vpara(\vx) = \frac{1}{2} \vx^T \Theta \vx
\]
More precisely, for a parameter $\vpara=[\theta_1\dotsc\theta_d]$,
the symmetric matrix $\Theta$ is constructed by assigning its entries:
\[
\Theta_{i,j} = \Theta_{j,i} = \theta_{((i-1)i/2 + j)} \qquad \text{for } i = 1\dotsc n, j = 1\dots i
\]
Further, we can derive the Lie derivative as:
\[
\gradlya_\vpara(\vx)\cdot \fbb(\vx) = (\Theta \vx) \cdot \fbb(\vx) = \vx^T \Theta \fbb(\vx)
\]
Following Section~\ref{subsec:learn}, a set of samples $S = \{(\xj_1, \yj_1)\dotsc (\xj_k, \yj_k)\}$ constrains $\Hyps_S$ by:
\[
\Hyps_S = \Hyps_0 \cap \left\{\vpara\in \Reals^d\ \Bigg|\ \bigwedge_{i=1}^k \frac{1}{2} \xj_i^T \Theta \xj_i \geq 0 \land \xj_i^T \Theta \yj_i \leq 0 \right\}
\]
We can find the analytical center $\vpara_{k+1} \in \Hyps_S$ by solving the following convex optimization problem:
\[
\vpara_{k+1} = \arg\max_{\vpara \in \Reals^d} 
\left(\begin{array}{c}
       \sum\limits_{i=1}^d \left(\ln(\frac{1}{2} + \theta_i) + \ln(\frac{1}{2} - \theta_i)\right) \\
       + \sum\limits_{i=1}^k \left(\ln(\frac{1}{2}\xj_i^T \Theta \xj_i) + \ln(-\xj_i^T \Theta \yj_i)\right)
\end{array} \right)
\]

\section{Details for Benchmarks}\label{appx:benchmarks}

\begin{table}[ht]
    \caption{Configurations for benchmarks.}\label{tab:exp-conf}
    \centering
    \begin{tabular}{l|c|c|r}
        Name     &$n$D& \ROI $\States$                 & $\lipbb$ \\
        \hline
        Van der Pol~\cite{zhou_neural_2022}
                 & 2D & $0.2 \leq \norm{\vx} \leq 1.2$ & 4.632 \\
        \hline
        Unicycle path~\cite{zhou_neural_2022}
                 & 2D & $0.1 \leq \norm{\vx} \leq 0.8$ & 62.171 \\
        \hline
        Inverted pendulum~\cite{zhou_neural_2022}
                 & 2D & $0.4 \leq \norm{\vx} \leq 4.0$ & 12752 \\
        \hline
        Stanley Controller~\cite{hoffmann_autonomous_2007}
                 & 2D & $10^{-3} \leq |e|\leq 2 \land 10^{-3} \leq|\psi|\leq\frac{\pi}{4}$      & 3.266
    \end{tabular}
    \vspace{5pt}

    \begin{tabular}{l|c|r|r|r}
                             &      & \multicolumn{3}{c}{Lipschitz bound $\lipbb$} \\
         \cline{3-5}
        Name                 & $n$D & $r=1$ & $r=5$ & $r=10$ \\
        \hline
        $\mathtt{nonpoly}_0$ &  2D  & 2.449 & 7.874 & 14.90\\
        $\mathtt{nonpoly}_1$ &  2D  & 5.477 & 112.7 & 448.1\\
        $\mathtt{nonpoly}_2$ &  3D  & 3.178 & 7.778 & 24.27\\
        $\mathtt{nonpoly}_3$ &  3D  & 3.564 & 40.61 & 317.2\\
    \end{tabular}
    \hspace{1em}
    \begin{tabular}{l|c|r|r|r}
                             &      & \multicolumn{3}{c}{Lipschitz bound $\lipbb$} \\
        \cline{3-5}         
        Name                 & $n$D & $r=1$ & $r=5$ & $r=10$ \\
        \hline
        $\mathtt{poly}_1$    &  3D  & 10.63 & 210.8 & 838.2 \\
        $\mathtt{poly}_2$    &  2D  & 3.464 & 75.02 & 300.0 \\
        $\mathtt{poly}_3$    &  2D  & 5.477 & 110.2 & 432.0 \\
        $\mathtt{poly}_4$    &  2D  & 7.382 & 3577  & 57063\\
    \end{tabular}
\end{table}

We report the configuration used for each benchmarks in \cref{tab:exp-conf}.
All benchmarks except for the Stanley Controller are differentiable.
We derive the regional Lipschitz bounds $\lipreg$ by first computing the Jacobian matrix $J_\fbb(\vx)$
and calculate the Frobenius norm $\norm{J_\fbb(\vx)}_F$.
For the domain $\Dom$ and a simplex region $\Reg$,
we further find an upper bound of the norm as the Lipschitz bound $\lipbb\geq\sup_{\vx\in\Dom} \norm{J_\fbb(\vx)}_F$ and the regional Lipschitz bound $\lipreg \geq \sup_{\vx\in\Reg} \norm{J_\fbb(\vx)}_F$.
Because the computation is not exact,
we may use an ever larger value for soundness.
We report the Lipschitz bound $\lipbb$ in the domain $\Dom$ in \cref{tab:exp-conf}.

\paragraph{Lipschitz Bound for Stanley Controller}
We consider the kinetic vehicle model from~\cite{hoffmann_autonomous_2007}.
We simplify the model with a constant velocity $v$ instead of $v(t)$ and denote the length of wheel base $l=a+b$.
We can manually derive a Lipschitz bound $\lipbb = \sqrt{(1+l^{-2})\cdot(v^2 + k^2)}$.
For our evaluation, we use $k = 0.45$, $l = 1.75$, and $v = 2.8$,
so $\lipbb\approx 3.266$.

\section{Comparison of Related Works}\label{appx:related}

\begin{table}[ht]
    \centering
    \caption{Comparison on components in \CEGIS of Lyapunov functions.}
    \label{tab:appx-related}
    \begin{footnotesize}
    \begin{tabular}{l|c|c|c|c|c|c}
     & Target    &         &          \\
     & Sys. Time & Learner & Verifier \\ \hline
\cite{ahmed_automated_2020,abate_formal_2021,abate_fossil_2021}
     & CT   & NN      & SMT      \\
\cite{chang_neural_2019}
     & CT   & NN      & SMT      \\
\cite{chen_learning_2021,chen_learning_2021-1}
     & DT   & SDP+CP  & MIQP     \\
\cite{berger_learning_2022,berger_counterexample-guided_2023}
     & CT   & CP      & LP       \\
\cite{masti_counter-example_2023}
     & DT   & SDP     & CP       \\
\cite{ravanbakhsh_counter-example_2015,ravanbakhsh_counterexample-guided_2015,ravanbakhsh_robust_2016,ravanbakhsh_learning_2019}
     & Both & LP      & SDP      \\
\cite{zhou_neural_2022}
     & CT   & NN      & SMT      \\
Ours & CT   & CP      & SMT      \\
    \end{tabular}
    \hspace{1em}
    \begin{tabular}{l}
    Acronyms: \\
    Continuous Time~(CT), \\
    Discrete Time~(DT), \\
    Convex Programming~(CP), \\
    Linear Programming~(LP), \\
    Semidefinite Programming~(SDP), \\
    Mixed Integer Quadratic Programming\\
    (MIQP), \\
    Neural Networks~(NN), \\
    Satisfiability Modulo Theories~(SMT)
    \end{tabular}
    \end{footnotesize}
\end{table}

\paragraph*{Choices on Approximations}
There are also different choices on what expressions in Condition~\eqref{cond:lya-stab} to approximate.
For instance, \cite{zhou_neural_2022} approximates the dynamics $\fbb$ with neural networks with error bound by universal approximation theory.
\cite{bobiti_automated-sampling-based_2018} computes a piecewise approximation of the Lie derivative $\liederlya$,
and it derives a condition using the approximation to determine when the Lie derivative must be negative definite.
\cite{masti_counter-example_2023} handles discrete time control affine systems with convergence.
\cite{samanipour_automated_2023} use Delaunay triangulation for piecewise affine dynamics.

\subsection{Comparison with~\cite{zhou_neural_2022}}\label{appx:subsec:compare-neurips2022}

We do not compare the computation time of the implementation of~\cite{zhou_neural_2022} with our approach because the comparison will be unfair under our setup and highly favors our approach.
This is based on the following three main reasons:
\begin{itemize}
\item The implementation of~\cite{zhou_neural_2022} obtains millions of evenly-spaced samples assuming a vectorized black-box function $\fbb$ for GPU acceleration. Under our setup, samples from a black-box system are obtained iteratively, and their approach suffers from the massive amount of samples.
\item The implementation of \cite{zhou_neural_2022} further synthesizes control Lyapunov functions and controllers that modify the behavior of the black-box system. Our approach focuses on certifying the stability and does not synthesize the controller.
\item Our hardware platform has better multi-core CPU for multiprocessing and less powerful GPU for Neural Network-based computation.
\end{itemize}

\end{document}